\documentclass[a4paper,UKenglish,cleveref, autoref, thm-restate]{lipics-v2021}

\title{Satisfiability of Context-free String Constraints with Subword-ordering and Transducers} 
\titlerunning{Context-free String Constraints with Subword-ordering and Transducers} 

\author{C. Aiswarya}{Chennai Mathematical Institute, India \and  CNRS, ReLaX, IRL 2000, India}{aiswarya@cmi.ac.in}{https://orcid.org/0000-0002-4878-7581}{}

\author{Soumodev Mal}{Chennai Mathematical Institute,  India}{soumodevmal@cmi.ac.in}{https://orcid.org/0000-0001-5054-5664}{}

\author{Prakash Saivasan}{Institute of Mathematical Sciences, HBNI, India \and CNRS, ReLaX, IRL 2000, India}{psaivasan@imsc.res.in}{https://orcid.org/0000-0001-5060-0117}{MATRICS GRANT (MTR/2022/000312)}

\authorrunning{C. Aiswarya, S. Mal, and P. Saivasan} 

\Copyright{C Aiswarya, Soumodev Mal, and Prakash Saivasan} 


\category{} 

\relatedversion{This is the full version of a paper accepted in STACS'24} 

 \acknowledgements{We thank Paul Gastin for helpful discussions.}

\nolinenumbers 

\hideLIPIcs  

\usepackage{todonotes}
\usepackage{enumitem}
\usepackage{caption}

\usepackage{tikz}
\usetikzlibrary{automata}
\tikzstyle{branch}=[fill,shape=circle,minimum size=3pt,inner sep=0pt]

\begin{document}

\maketitle

\newcommand{\twonexpc} {{\sc 2NExptime}-complete}
\newcommand{\twonexph} {{\sc 2NExptime}-hard}
\newcommand{\twonexp} {{\sc 2NExptime}}
\newcommand{\nexpc} {{\sc NExptime}-complete}
\newcommand{\nexph} {{\sc NExptime}-hard}
\newcommand{\nexp} {{\sc NExptime}}

\pgfdeclarelayer{bg}    
\pgfsetlayers{bg,main}

\begin{abstract}
	We study the satisfiability of string constraints where context-free membership constraints may be imposed on variables. Additionally a variable may be constrained to be a subword of a word obtained by shuffling variables and their transductions. The satisfiability problem is known to be undecidable even without rational transductions. It is known to be \nexpc\ without transductions, if the subword relations between variables do not have a cyclic dependency between them. We show that the satisfiability problem stays decidable in this fragment even when rational transductions are added. It is \twonexpc\  with context-free membership, and \nexpc\  with only regular membership. For the lower bound we prove a technical lemma that is of independent interest: The length of the shortest word in the intersection of a pushdown automaton (of size $\mathcal{O}(n)$) and $n$ finite-state automata (each of size $\mathcal{O}(n)$)  can be double exponential in $n$.
	
\end{abstract}

\section{Introduction}

The theory of strings has always been an important and  active area of research for long. In fact, as Hilbert notes, it is the very foundation of mathematical logic itself \cite{frege-to-godel,CorcoranFM74}. The recent successes in employing the theory for practical verification has only re-iterated its importance.
The study of the theory of string constraints dates back to Tarski and Hermes \cite{Tarski1935, Quine1939HermesHS}, who in 1933 provided the axiomatic foundation for it. There have been several other advancements of string theories since then, some of the notable ones include \cite{CorcoranFM74,quine_1946, Matiyasevich, GSMakanin, Plandowski99, BuchiSenger1988}. In 1977, Makanin  studied the algorithmic aspect of the word equations (equation involving concatenation and equality) and showed that the satisfiability problem is decidable \cite{GSMakanin}. The complexity for this problem was improved in \cite{Plandowski99}. Despite receiving much attention, the theory of strings  has long standing unsolved open problems, indicating the intrinsic difficult nature of the theory.

 One important aspect of the study here is the satisfiability of string constraints. The question here asks whether it is possible to  assign a word to each variable such that the given set of string constraints is satisfied. The constraints themselves can be either relational,  which relate variables or membership,  that define the domain for each variable. 

In the recent years, the constraint satisfaction problem of strings (CSPS) has received much attention from verification community due to its usefulness in modeling and reasoning about programs. This problem has particularly been useful in verifying web services \cite{HagueLO15} and database applications from injection attacks \cite{Amadini23}.  In such attacks, the attacker constructs an input string in such a way that the underlying semantics of the interpretation is changed. The CSPS, and more importantly its implementations in solvers \cite{CVC416,Trau2018,Z3str2017,Hampi13,KiezunGHEG19,Day22,KanLRS22} have provided the much needed power to model and verify programs for such vulnerabilities. This in turn has directed the study  to explore the boundaries of solvability.

 However one impediment for this has been the theoretical limitation. For instance, with respect to word equations, adding a transducer renders the model undecidable. Similarly introducing membership in context free language also renders the model undecidable (see \cite{DayGHMN2018}, \cite{GaneshMSR12} for more details).  Despite this, there have been several advancements in this regard \cite{HagueLin,LinB16,HolikJLRV18,ChenHLRW19,ChenHHHLRW20,ChenCHLW18,AbdullaACHRRS14,AbdullaADHJ19,FigueiraJL22}.

The context-free membership constraints are particularly useful feature to have since checking vulnerabilities include checking for programs, that are inherently context-free,  masquerading as string queries. In \cite{AiswaryaMS22}, the authors provided first such model that could handle context-free membership queries and yet has decidability for CSPS, under some restrictions. They showed that if every relational constraint has sub-word relation instead of equality and assuming an \emph{acyclicity} restriction,  the satisfaction problem is {\sc NExptime complete}. In fact, the authors in their model include  a more powerful shuffle operator against the usual concatenation.  Further they show that the complexity of the satisfiability problem when only regular membership is involved is also the same i.e, {\sc NExptime complete}. They also provide an interesting connection of their model with lossy channel systems that include pushdown automata.

Yet another feature  in string solvers that has been much desired is that of transductions. As noted in \cite{HolikJLRV18, ChenHHHLRW20}, most modern applications, especially browsers include implicit transductions that mutates the input string. To verify such applications, one also needs the power of transductions. There have been very few successful attempt towards decidability of string constraints that involve transductions, some of them being  \cite{HolikJLRV18, AbdullaADHJ19, ChenHHHLRW20,HagueLin}.

We investigate string constraints when sub-word ordering, context-free membership and transducers are involved. 
Unfortunately, in its full generality this problem is undecidable. However we show that imposing the same acyclicity restriction as in \cite{AiswaryaMS22} gives decidability under this setting. This extends the decidability result of \cite{AiswaryaMS22} to include transductions.

In \cite{AiswaryaMS22} the satisfiability of the acyclic variant of the string constraints without transducers was shown to be inter-reducible with the control-state reachability problem of acyclic networks of pushdown systems communicating over lossy fifo channels. They showed that both these problems are \nexpc. In our setting, with the additional feature of transductions, we can enrich the model of communicating pushdown to allow transductions to be sent in the channels. Such transductions naturally model encoders such as error correcting codes or injection of noise.  

 We show that, when only regular membership is allowed, adding transductions do not alter the complexity. It is still \nexpc. Interestingly when context-free membership is involved, it becomes \twonexpc.

 Our \twonexp lower bound argument relies on a new technique that is of independent interest. In fact, we show that we can count exactly $2^{2^n}$ using one pushdown automaton with a binary stack alphabet and 3 states,  and $n$ finite state automata each of size $\mathcal{O}(n)$. Along the way we also show that 
 1) we can count exactly $2^n$ using a pushdown automaton with $\mathcal{O}(n)$ states and a binary stack alphabet, and 2) we can count exactly $2^n$ using $n$ finite state automata each of size  $\mathcal{O}(n)$.

As an application of this, we obtain a tight bound on the size of the smallest DFA of the downward closure, upward closure and the Parikh image closure of the intersection  language of $n$ finite state automata, each of size $\mathcal{O}(n)$. This size is ${\Theta}(2^n)$. Likewise, the size of  the smallest DFA  of the downward, upward and Parikh image closure for the intersection of language of $n$ finite state automata  with the language of a pushdown automaton, each of size  $\mathcal{O}(n)$ is  ${\Theta}(2^{2^n})$. 
\paragraph*{Related work}
Apart from the work mentioned in the introduction, there are several other work on string constraints. In \cite{ChenCHLW18}, the authors consider word equations equipped with replace all function and show decidability for the acyclic fragment. 

In \cite{AbdullaMFC17}, the authors develop an uniform framework to decide the satisfiability and unsatisfiability of string constraints  based on identifying patterns. In \cite{AbdullaACDHHTWY21}, the authors consider string constraints extended with negation and show how to solve them.
In \cite{ChenHLRW19}, the authors provide a semantic restriction on string manipulating programs that guarantees decidability for checking path feasibility. In \cite{AbdullaADK20}, the authors study the problem of regular separability of the language of two word equations. In \cite{DayGGM23}, the authors compare the expressive power of the logical theories built around word equations.

In \cite{HolikJLRV18}, word equations  with equality, transducers and regular membership is considered. This problem in full generality is immediately  undecidable. The authors consider a straight line fragment and show that the satisfiability problem is {\sc Expspace complete}. In \cite{ChenHLRW19}, the authors investigated the decidability of string constraints in the presence of regular membership constraints, \texttt{replaceAll} operator involving regular expressions and straight line restriction. In \cite{abs211104298}, the authors consider a stronger match and replace operator and show decidability.  In \cite{AbdullaADHJ19}, word equations  with equality, transducers, length constraints and regular membership is considered and a  chain free fragment of it was shown to be decidable. The authors show that the chain-free fragment of the satisfiability problem in this setting is decidable. 

All of these work consider word equation (uses equality for comparison) in the model, our work uses subword ordering as the comparison operator.  Further more, none of the work mentioned above considers context free membership constraints. In \cite{AiswaryaMS22}, subword ordering and context free membership is considered, where as it does not include transductions.

Apart from these, there are several approaches which attempts to solve the problem from a practical perspective, some of them being \cite{ParoshMYDLAP,AbdullaACHRRS15,AbdullaACDHHTWY21,AbdullaACHRRS14, Z3str2017, BerzishDGKMMN23, HolikJLRV18}.

\section{Preliminaries}\label{sec:preliminaries}
\newcommand{\Nat}{\mathbb{N}}
\newcommand{\Natz}{\Nat_0}
\newcommand{\nset}[1]{[#1]}
\newcommand{\multiset}[1]{\{\!\!\{#1\}\!\!\}}
\newcommand{\msetunion}{\oplus}
\newcommand{\sizeof}[1] {|#1|}

\newcommand{\alphabet}{\Sigma}
\newcommand{\alphabeteps}{\alphabet_\epsilon}
\newcommand{\subword} {\preceq}
\newcommand{\len}[1]{\textsf{len}(#1)}
\newcommand{\pos}[1]{\textsf{pos}(#1)}
\newcommand{\proj}[2]{{#1}_{\downarrow{#2}}}
\newcommand{\shuffle}[1]{\mathsf{Shuffle}(#1)}
\newcommand{\functionset}{\mathcal F}
\newcommand{\terms}[2]{#1(#2)}

\paragraph*{  Sets, Multisets, Functions } We denote the set of natural number $\{1,2, \dots \}$ by $\Nat$. For $n \in \Nat$, we denote by $\nset{n}$ the set of natural numbers up to $n$: $\{1, 2, \dots, n\}$. Let $\Natz $ denote the set $\{0, 1, 2, \dots\}$. That is, $\Natz = \{0\} \cup \Nat$.

Let $S$ be any set. A \textit{multiset} $X$ of $S$ assigns a multiplicity $X(s) \in \Natz$ to each element $s\in S$. We say that $s\in X$ if $X(s)>0$.
For a usual subset $X$, the multiplicity $X(s) \in \{0,1\}$. A multiset $X$  may also be written as $\multiset{s_1, s_2, \dots}$, by listing each element $s$, $X(s)$ many times. The set of all multisets of $S$ is denoted $\Natz^S$, and the set of all usual subsets of $S$ is denoted by $2^S$. The size of a multiset $X$, denoted $\sizeof{X}$ is the sum of the multiplicities of the elements. That is, $\sizeof{X} = \sum_{s \in X} X(s)$.

\paragraph*{ Word, Subword, Shuffle, Projection}
Let $\alphabet$ be an alphabet. $\alphabet^\ast$ denotes the set of all words over $\alphabet$, $\epsilon$ denotes the empty word, and  $\alphabeteps = \alphabet \cup \{\epsilon\}$.  For a word $w = a_1a_2\dots a_n\in \alphabet^\ast$, we denote by $\len{w}$, the \textit{length} of $w$ ($\len{w} = n$) and by $w[i]$ its $i$th letter $a_i$. The set of positions of $w$ is denoted $\pos{w}$.  That is, $\pos{w} = \nset{\len{w}}$. For $Y \subseteq \pos{w}$, we denote by $\proj{w}{Y}$ the \textit{projection} of $w$ to the positions in $Y$.  If 
$Y = \{i_1, i_2, \dots i_m\}$ with $0 < i_1 < i_2 < \dots < i_m \le n$, then $\proj{w}{Y} = a_{i_1}a_{i_2} \dots a_{i_m}$.
For $u, v \in  \alphabet^\ast$, we say $u$ is a \textit{(scattered) subword} of $v$, denoted $u \subword v$,  if  there is $Y \subseteq \pos{v}$ such that $u = \proj{v}{Y}$. In this case we say $v$ is a \textit{superword} of $u$. Let $\alphabet' \subseteq \alphabet$ be a sub-alphabet and let $w \in \alphabet^\ast$. Projection of $w$ to $\alphabet'$, denoted $\proj{w}{\alphabet'}$, is defined to be $\proj{w}{Y}$ where $Y = \{i \mid w[i] \in \alphabet'\}$.

Let $X$ be a finite multiset of words from $\alphabet^\ast$ given by $X =  \multiset{w_1, \dots w_n}$. We define the \textit{shuffle} of $X$, denoted $\shuffle{X}$ to be the set $\{w \mid  $ there are $ Y_1, Y_2 \dots Y_n \subseteq \pos{w} $ forming a partition of $ \pos{w} $ and $w_i = \proj{w}{Y_i} $ for all $ i \in \nset{n}  \}$.

\newcommand{\pstates}{\textsf{States}}
\newcommand{\stackalphabet}{\Gamma}
\newcommand{\ptrans}{\textsf{Trans}}
\newcommand{\pinstate}{s_\text{in}}
\newcommand{\pfinstates}{F}
\newcommand{\lab}{\textsf{label}}
\newcommand{\movesto}[1]{\xrightarrow{#1}}
 \newcommand{\Ops}{\textsf{Ops}}
\newcommand{\push}{\textsf{push}}
\newcommand{\pop}{\textsf{pop}}
\newcommand{\nop}{\textsf{nop}}
 \newcommand{\PDA}{\mathsf{PDA}}
  \newcommand{\NFA}{\mathsf{NFA}}
   \newcommand{\TRANSD}{\mathsf{TRANSD}}
\newcommand{\out}{\textsf{out}}
\newcommand{\op}{\textsf{op}}
\newcommand{\transducerset}{\mathcal T}
\newcommand{\transducer}{T}
\newcommand{\idtransducer}{T_\text{id}}
\newcommand{\statesize}[1]{\textsf{state-size}(#1)}
\newcommand{\occ}[1]{n_#1}
\newcommand{\TRSET}[1]{\textsc{trset}(#1)}
\newcommand{\AUTSET}[1]{\textsc{autset}(#1)}
\paragraph*{Finite-state automaton, Transducers, Pushdown Automaton}

A (nondeterministic) \textit{finite-state automaton} (NFA) over an alphabet $\alphabet$ is given by a tuple $A = (\pstates, \ptrans, \pinstate, \pfinstates)$ where $\pstates$ is the finite set of states, $\ptrans \subseteq \pstates \times \alphabeteps \times \pstates$ is the set of transitions, $\pinstate \in \pstates$ is the initial state, and $\pfinstates \subseteq \pstates$ is the set of final/accepting states. We write $s \movesto{t}s'$ for some $t \in \ptrans$ if $t$ is of the form $(s, a, s')$.  Define the homomorphism $\lab: \ptrans^\ast \to \alphabet^\ast$ given by $\lab((s,a, s')) = a$. The \textit{language} of an NFA $A$, denoted $L(A)$ is given by $L(A) = \{w \mid w = \lab(t_1t_2\dots t_n)$ and $\pinstate \movesto{t_1} s_1 \movesto{t_2} s_2 \dots s_{n-1} \movesto{t_n} s_n$ with $s_n \in \pfinstates\}$. 

\noindent \begin{minipage}{0.45\textwidth}
A  \textit{transducer} from $\alphabet^\ast$ to $\alphabet^\ast$ is a tuple $\transducer  = (\pstates, \ptrans, \pinstate, \pfinstates, \out)$ where $A = (\pstates, \ptrans, \pinstate, \pfinstates)$ is an NFA, and $\out:\ptrans \to \alphabet^\ast $ defines the outputs on each transition. The function $\out$ defines a homomorphism $\out: \ptrans^\ast \to \alphabet^\ast$. The \textit{relation} $R \subseteq \alphabet^\ast \times \alphabet^\ast$ recognized by $\transducer$, denoted $R(\transducer)$ is given by $\{(u,v) \mid u = \lab(t_1t_2\dots t_n), v = \out(t_1t_2\dots t_n)$ and $\pinstate \movesto{t_1} s_1 \movesto{t_2} s_2 \dots s_{n-1} \movesto{t_n} s_n$ with $s_n \in \pfinstates\}$. The equality relation is realised by a transducer $\idtransducer$. A transducer is depicted in Figure~\ref{fig:exTransducer}.
\end{minipage}\hfil
\begin{minipage}{0.475\textwidth}
\centering
	\begin{tikzpicture}[node distance=2.25cm, initial text=,>=latex, thick]
	
	\node[state, initial above, ] (2) {$2$};
	\node[state, accepting, right of=2,] (3) {$3$};
	\node[state, accepting, left of=2] (1) {$1$};
	
	\draw[->] 
	(2) edge node[above] {$\epsilon/\epsilon$} (1)
	(2) edge node[above] {$\epsilon/\epsilon$} (3)
	(1) edge[loop above] node[above] {$b/bb$} (1)
	(1) edge[loop below] node[below] {$a/a$} (1)
	(3) edge[loop above] node[above] {$b/b$} (3)
	(3) edge[loop below] node[below] {$a/aa$} (3)
	;
\end{tikzpicture}

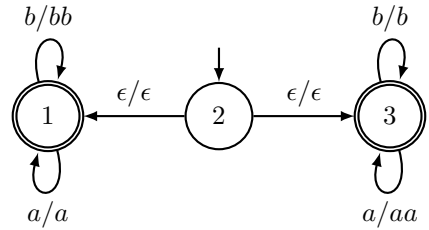
\captionof{figure}{A transducer. Here the label $x/y$ on a transition $t$ indicates that $\lab(t) = x$ and $\out(t) = y$. It nondeterministically chooses to duplicates $a$s leaving $b$s as such, or duplicates $b$s leaving $a$s as such.} \label{fig:exTransducer}
\end{minipage}

A \textit{pushdown automaton} over  $\alphabet$ is given by a tuple $P = (\pstates, \ptrans, \pinstate, \pfinstates, \op, \stackalphabet)$ where  $A = (\pstates, \ptrans, \pinstate, \pfinstates)$ is an NFA,  $\stackalphabet$ is the finite set of stack symbols, and $\op:\ptrans \to \Ops $ defines the stack operation of each transition, where $\Ops = \{\push(\gamma) \mid \gamma \in \stackalphabet\} \cup \{\pop(\gamma) \mid \gamma \in \stackalphabet\} \cup \{\nop\}$. When depicting the pushdown automaton pictorially, we represent a transition $t =(s, a, s')$ as $s \xrightarrow{a\mid op(t)} s'$. 
When $\op(t) = \nop$, we may simply write  $s \xrightarrow{a} s'$. Further if $a = \epsilon$ then we may write it as $s \xrightarrow{op(t)} s'$. 
A configuration of a PDA is a pair $(s,w) \in \pstates \times \stackalphabet^\ast$, indicating the current state and the stack contents. For two configurations $(s,w)$ and $(s', w')$ we write $(s,w) \movesto{t}(s',w')$ for some $t \in \ptrans$ if $t$ is of the form $(s, a, s')$  and 1)  $\op(t) = \push(\gamma)$ and $w' = \gamma\cdot w$, or 2) $\op(t) =  \pop(\gamma)$ and $w = \gamma \cdot w'$, or 3) $\op(t) = \nop$ and $w = w'$. The \textit{language} of a PDA P, denoted $L(P)$ is given by $L(P) = \{w \mid w = \lab(t_1t_2\dots t_n)$ and $(\pinstate,\epsilon) \movesto{t_1} (s_1,w_1) \movesto{t_2} (s_2,w_2) \dots (s_{n-1},w_{n-1}) \movesto{t_n} (s_n, \epsilon)$ with $s_n \in \pfinstates\}$.

The set of all NFA / transducers / PDA over the alphabet $\alphabet$ is denoted $\NFA(\alphabet)$ / $\TRANSD(\alphabet)$ / $\PDA(\alphabet)$.  A language $L \subseteq \alphabet^\ast$ is said to be context-free  (resp. regular) if there is a PDA (resp. NFA) $A$ such that $L = L(A)$. A relation $R \subseteq \alphabet^\ast \times \alphabet^\ast$ is said to be rational if it is recognized by some transducer $\transducer$. 

Given an NFA $A$ (resp. transducer $T$), its number of states is denoted by $\statesize{A}$ (resp. $\statesize{T})$. Given a PDA $P$ by $\statesize{P}$ we denote the sum of the number of states and number of stack symbols. That is $\statesize{P} = \sizeof{\pstates} + \sizeof{\stackalphabet}$.

\section{String constraints}\label{sec:stringConstraints}
\newcommand{\variableset}{\mathcal V}
\newcommand{\tuple}[1]{(#1)}
 \newcommand{\Mem}{\mathsf{Mem}}
 \newcommand{\Rel}{\mathsf{Rel}}
 \newcommand{\assignment}{\sigma}
 \newcommand{\multiplicity}{\textsf{multiplicity}}
 
 A string constraint  over a set of variables $\variableset$ and an alphabet $\alphabet$ is given by a set of membership constraints and a set of subword ordering constraints. The membership constraint is given by associating a pushdown automaton to each variable, indicating that the word assigned to the variable must belong to the language of the pushdown automaton.  A subword order constraint is given by a pair $(x, Y)$ where $x \in \variableset$ and $Y$ is a finite multiset over $ \variableset \times \TRANSD(\alphabet) $. 
 
  For example, the constraint $(x,  {\multiset{(y, \transducer_1), (y, \transducer_2), (y, \transducer_2), (z, \transducer_2)}})$ means that the words assigned to $x$, $y$ and $z$, say $w_x$, $w_y$ and $w_z$ respectively, must satisfy $w_x \subword w$ for some $w \in \shuffle{\multiset{u_1, u_2, u_3, u_4}}$, where $(w_y, u_1) \in R(\transducer_1)$, $(w_y, u_2) \in R(\transducer_2)$, $(w_y, u_3) \in R(\transducer_2)$, and $(w_z, u_4) \in  R(\transducer_2)$. Note that the transducers can be identity in which case the input and the output are the same. For instance,  if $\transducer_1 = \idtransducer$ then $u_1$ must be same as $w_y$. 
 
 We sometimes denote the constraint $(x,Y)$ by  $x \subword \shuffle{Y}$. Abusing notation, we may  write a pair $(x, T) \in  \variableset \times \TRANSD(\alphabet) $ as $T(x)$. If $Y = \multiset{(y, T)}$ (i.e., a singleton), then we  may simply   write $x \subword T(y)$ instead of $x \subword \shuffle{\multiset{T(y)}} $. Further, we may simply write $x$ for $(x, \idtransducer)$.  For instance, $(x, \multiset{(y, \idtransducer)})$ may be also written as $x \subword y$.

 \begin{definition}
 	A string constraint $C$ is a tuple $\sc = \tuple{\alphabet, \variableset,  \Mem, \Rel}$ where $\Mem : \variableset \to \PDA(\alphabet) \cup \NFA(\alphabet)$ assigns a PDA or an NFA to each variable, and $\Rel \subseteq \variableset \times \Natz^{  \variableset \times \TRANSD(\alphabet) }$  is a finite set of subword-order constraints. 
 	\end{definition}

 We denote by $\TRSET{C}$  the finite set of  transducers occurring in the string constraint $C$. That is, $\TRSET{C} = \{T \mid \exists (x,Y) \in \Rel, y \in \variableset, (y, T) \in Y\}$.  Similarly, $\AUTSET{C}$ is the finite set of PDA/NFA occurring in $C$. That is, $\AUTSET{C} = \{\Mem(x) \mid x \in \variableset\}$. A string constraint is \textit{regular} if for every $v \in \variableset$, $\Mem(v) $ is an NFA, (equivalently, if $\AUTSET{C} \subseteq \NFA(\alphabet)$). An important parameter for our complexity considerations will be the number of times a variable is used in the  right hand side (RHS). We denote it by $\multiplicity_C(x) = \sum_{T \in \TRSET{C}, (y,Y) \in \Rel} Y((x,T))$. We omit the subscript and simply write $\multiplicity(x)$ when $C$ is clear from the context.
 
 \begin{definition} \label{def:satassign}
 	A string constraint $C$ is \emph{satisfiable} if there exists an assignment $\assignment: \variableset \to \alphabet^\ast$ that satisfies every membership  and relational constraints in $C$ --- that is,
 	\begin{enumerate}
 		\item $\assignment(v) \in L(\Mem(v))$ for all $v \in \variableset$
 		\item\label{cond:rel-sat} For every $(x, Y) \in \Rel$, if $Y = \multiset{(y_1, T_1), ( y_2, T_2), \dots ( y_n, T_n)} $, then there are words $u_1, u_2, \dots , u_n$ such that $(\assignment(y_i), u_i) \in R(T_i)$ for each $i \in \{1, 2, \dots, n\}$, and there is a word  $w \in \shuffle{\multiset{u_1, u_2, \dots, u_n }}$ such that $\assignment(x) \subword w$ . 
 	\end{enumerate}
 Such an assignment $\assignment$ is called a \emph{satisfying assignment}. 
 	\end{definition}

\begin{example}  Consider a string constraint on two variables $x$ and $y$. The membership constraints are as follows. $\Mem(x)$  is an NFA for   $\{ababab\}$, and $\Mem({y})$  is an NFA for   $\{ab\}$. There is only one relational constraint:   $x \subword \shuffle{\multiset{(y, T)(y, T)}}$, where $T$ is the transducer defined in Figure~\ref{fig:exTransducer}. This string constraint is satisfiable. 
	\end{example}

\begin{definition}(Satisfiability Problem for String Constraints)\\	
		\begin{tabular}{c p{11.5cm}}	
		Input : & A string constraint $C$.\\
		Question:& Is $C$ satisfiable ?
	\end{tabular}
	\end{definition}

 The satisfiability problem is undecidable already for regular string constraints without transducers (or, equivalently, when only $\idtransducer$ is allowed) \cite{AiswaryaMS22}. To circumvent undecidability, acyclic fragment of string constraints were considered in \cite{AiswaryaMS22}. Formally, let $x< y$ if $(x, Y) \in \Rel$ with $(y, \transducer) \in Y$ for some transducer $\transducer$. The string constraint is acyclic if $<$ is acyclic. For the acyclic fragment without transducers, satisfiability was shown in \cite{AiswaryaMS22}  to be \nexpc. The lower bound already holds for regular acyclic string constraints without transducers.

We study the satisfiability problem for acyclic string constraints in the presence of transducers. Our main results are:
\begin{theorem}\label{thm:2nexpcCFL}
	Satisfiability problem for acyclic context-free string constraints with transducers is \twonexpc.
	\end{theorem}
\begin{theorem}\label{thm:nexpcReg}
	Satisfiability problem for acyclic regular string constraints with transducers is \nexpc.
\end{theorem}

\begin{remark} Our result shows an interesting contrast with string equations (with equality instead of subword order in relational constraints). Satisfiability of string equations (with concatenation, no shuffle) is decidable,  when regular membership constraints are allowed. Adding transducers on top however render the satisfiability undecidable. In our setting, where subword order is used instead of equality, 
	adding transducers to the acyclic fragment retains decidability. 
	\end{remark}
\begin{remark}
Without transducers,  regular and context-free string constraints have the same complexity. In the presence of transducers they are in different complexity classes.
	\end{remark}
\begin{remark}
	It was shown in \cite{AiswaryaMS22} that concatenation can be expressed by shuffle. This simulation is only linear and furthermore it preserves acyclicity. Thus our complexity upper bounds already hold for string constraints which uses the more popular concatenation operation instead of shuffle. Interestingly, the lower bounds in Theorem~\ref{thm:2nexpcCFL} and Theorem~\ref{thm:nexpcReg} already hold for the variant without shuffle. 
	\end{remark}
 In Section~\ref{sec:2nexp-lower-bound} and Section~\ref{sec:2nexptimeUpperBound} we prove the lower bound  and upper bound claimed in Theorem~\ref{thm:2nexpcCFL} respectively. The proof of Theorem~\ref{thm:nexpcReg} is given in Section~\ref{sec:reg}.  In Section~\ref{sec:discussions}, we discuss some implications of our results and conclude.

\section{2NEXPTIME Hardness}\label{sec:2nexp-lower-bound}

We prove the hardness by giving a reduction from a bounded variant of the PCP problem that is \twonexpc.

\subsection{(Double-exponentially) Bounded PCP problem}

In this decidable variant of the PCP problem, we are also given a parameter $\ell$ as part of  the input in unary, and we ask whether there is a solution of length $2^{2^\ell}$. Formally the problem is stated as follows. 

\begin{definition}(Double-exponentially) Bounded PCP problem (2eBPCP).\\
	\begin{tabular}{c p{11.5cm}}	
		Input : & $(\Sigma_1, \Sigma_2, f, g, \ell)$ where $\Sigma_1$ and $\Sigma_2$ are two disjoint finite alphabets, $f, g : \Sigma_1 \to \Sigma_2^\ast$ are two functions which naturally extend to a homomorphism from $\Sigma_1^\ast \to \Sigma_2^\ast$, and $\ell \in \mathbb N$ is a natural number.\\
		Question:& Is there a word  $w \in \Sigma_1^+$ with $|f(w)| = 2^{2^\ell} $ and  $f(w) = g(w)$ ?
	\end{tabular}	
\end{definition}

The above problem is \twonexpc.  If the problem asked for the length of $f(w)$ to be $\ell$, it would be NP-complete \cite{ComputersIntractabilityGareyJohnson1990}, and if it was $2^\ell$ it would be \nexpc\ \cite{AiswaryaMS22}.

\begin{theorem}
	(Double-exponentially) Bounded PCP problem is  \twonexpc.
\end{theorem}\label{thm:boundedPCP}
\begin{proof}
	 \newcommand{\poly}{\textsf{poly}}
	\newcommand{\TM}{\textsf{TM}}
	\newcommand{\uset}{\mathcal{U}}
	\newcommand{\vset}{\mathcal{V}}
	\newcommand{\States}{{Q}}
	\newcommand{\PCP}{\mathcal{P}}
	\newcommand{\mpcp}{bounded-MPCP}
	
	Clearly  the 2eBPCP problem is in \twonexp as we can guess a solution of the appropriate size and verify it. 
	
	For the hardness, we first consider a modified version of bounded PCP called the bounded MPCP and show that this problem is  \twonexph. 
	Out proof strategy is similar to that in \cite{AiswaryaMS22}. Following the technique provided in \cite{HopcroftBook}, it is easy  to reduce this problem to bounded PCP. The bounded MPCP asks, given two equi-dimensional vector of words $\uset = (u_1, \cdots, u_n )$ and $\vset = (v_1, \cdots, v_n )$ over an alphabet $ \alphabet$ and an integer $\ell \in \Nat$, whether there is a sequence $i_1, \cdots i_k \in \nset{n}^+$ such that $u_1 \cdot u_{i_1} \cdot u_{i_2} \cdots u_{i_k} = v_1 \cdot v_{i_1} \cdot v_{i_2} \cdots v_{i_k}$ and $\sizeof{u_1 \cdot u_{i_1} \cdot u_{i_2} \cdots u_{i_k}} = \sizeof{v_1 \cdot v_{i_1} \cdot v_{i_2} \cdots v_{i_k}} = 2^{2^\ell}$. Notice that here we require that the solution start from a designated initial index. Going from bounded MPCP to bounded PCP requires only a linear blowup.

	In order to give the reduction, we fix a non-deterministic  Turing machine $\TM$ over the alphabet $\{0,1\}$ of size $n$ and an input $w$ and show how to construct an bounded MPCP $\PCP$ and an $\ell$ instance such that 
	the $\TM$ has an accepting run on $w$ of size at most $2^{2^{\poly(\sizeof{w},n)}}$ for some polynomial $\poly()$ if and only if $\PCP$ has a solution of size exactly $2^{2^{\ell}}$. Further more, the size of $\PCP$ and $\ell$ will only be quadratically dependent on $\poly()$,  $n$ and $\len{w}$.

	We first recall the construction that reduces an unrestricted $\TM$ to unrestricted MPCP from \cite{HopcroftBook}. 
	We  assume that the transitions of the Turing machine $\delta \subseteq \States \times \Sigma \times \States \times \Sigma \times \{ R,L\}$, where $\States$ are the set of states of the Turing machine and $\{R,L\}$ denotes the directions of the head movement i.e.  right, left. Further we will assume that $q_0$ is the start state of our Turing machine and that $F$ is the set of final states.

	The required PCP instance is $\PCP$ and is given  below. Here $\uset$ ( $\vset$ ) are obtained by projecting to the first (second) component of the pairs given below.
	
	\begin{eqnarray*}
		\{(\#,\#q_0w\#),(0,0),(1,1),(\#,\#) \} & \cup  & 
		\{(qx,yp) \mid (q,x,p,y,R) \in \delta, q \notin F \}  \\
		\{(xqy,pxz) \mid (q,y,p,z,L) \in \delta , q \notin F\} & \cup & 
		\{ (qx,q), (xq,q) \mid q \in F, x \in \Sigma \} \\
		\{ (q\#\#,\#) \mid q \in F \}  &
	\end{eqnarray*}

	The following lemma provides us with the required correctness and the value of $\ell$.
	
	\begin{lemma}\label{lem:apx:TMpcp}
		The Turing machine $\TM$ has an accepting run of size $m$ on $w$ then the  MPCP instance $\PCP$ has a solution of size $c \times m^2$ for some constant $c \in \Nat$.
	\end{lemma}
	\begin{proof}[Proof idea:]
		Assume that  there is a computation of the Turing machine of the form $C_1 {\movesto{}} C_2 \movesto{} \cdots \movesto{} C_n$, where each $C_i$ is a configuration of the Turing machine. We will assume that each configuration is of size at most $n$. In this case, there is a partial solution to the MPCP instance of the form  $\#C_1\#C_2\# \cdots\#C_n\#, \#C_1\#C_2\# \cdots\#C_{n-1}$.
		Suppose $C_n$ contains a final state then the completion is done by reducing one letter from the final configuration at a time using the pairs $(qx,q), (xq,q)$. We refer to these as the completion suffix. Notice that the completion suffix is of size  $ |C_n| \times (|C_n| -1)$. Hence we have a \mpcp\ solution of size $2 \times n^2$.
	\end{proof}
	The other direction  is as in the following lemma and is easy to see. 
	\begin{lemma}\label{lem:apx:pcpTM}
		If the MPCP instance $\PCP$ that we have constructed has a solution of size $m$, then The Turing machine $\TM$ has an accepting run of size at most $m$ on $w$ 
	\end{lemma}
	
	Now suppose we want to find if $\TM$ has an accepting run on $w$ of size at most $2^{2^{\poly(\sizeof{w},n)}}$, we let $\ell = c \times (\poly(\sizeof{w},n))^2$. Then by Lemmas \ref{lem:apx:pcpTM} and \ref{lem:apx:TMpcp}, we have 
	$\TM$ has an accepting run on $w$ of size at most $2^{2^{\poly(\sizeof{w},n)}}$ if and only if the \mpcp\ instance $\PCP$, $\ell$ has a solution.
\end{proof}

\subsection{Towards a reduction}

Our idea is to use 4 variables  $ x_1, x_2, x_f, x_g $. The membership constraint for $x_f$ is a PDA for the language $L_f = \{w \cdot\#^\ast \cdot f(w^r) \mid w \in \Sigma_1^\ast\}$, and that for  $x_g$ is a PDA for the language $L_g = \{w\cdot \#^\ast \cdot  g(w^r) \mid w \in \Sigma_1^\ast\}$. Recall that $w^r$  denotes the reverse of  $w$, and $\#$ is a special symbol not in $\Sigma_1$ or $\Sigma_2$. Suppose $x_1$ and $x_2$ are constrained to the language $\Sigma_1^m \#^n\Sigma_2^{2^{2^\ell}}$ such that $m+n = 2^{2^\ell}$, by  polynomial-sized constraints. Then with the relational constraints 1) $x_1 \subword x_f$ 2) $x_f \subword x_g$ and 3) $x_g \subword x_2$, we will achieve our reduction. Recall that $x \subword y$ is a short hand for $x \subword \shuffle{\multiset{(y, \idtransducer)}}$.  Indeed these constraints are satisfiable  if and only if the 2eBPCP has a solution. 

Notice that our constraints for $x_1$ and $x_2$ requires \textit{counting} exactly $2^{2^\ell}$. This is not possible with a polynomial-sized PDA.  In the above paragraph we did not use transducers either. Without transducers, the satisfiability problem of string constraints is not \twonexph, it is indeed in \nexp \cite{AiswaryaMS22}. 

However, with the help of $\ell$ many transducers (or FSA) of size $\mathcal{O}(\ell)$ we can have a PDA that counts $2^{2^\ell}$.  We will describe this technique with PDA and DFA in the next subsection, and in the following subsection using this idea, we complete the reduction.

\subsection{Counting $2^{2^\ell}$ using one PDA and $\ell$ DFA}
\newcommand{\bzero}{{\color{blue} 0}}
\newcommand{\bone}{{\color{blue} 1}}
\newcommand{\mr}{{\color{blue} {\rightarrow}}}
\newcommand{\ml}{{\color{blue} {\leftarrow}}}
\newcommand{\rzero}{{\color{red} 0}}
\newcommand{\rone}{{\color{red} 1}}

\newcommand{\modeswitch}{\textsf{switch}}
\newcommand{\modepush}{\textsf{push}}
\newcommand{\modepop}{\textsf{pop}}
\newcommand{\statepush}[1]{\langle #1, \modepush \rangle}
\newcommand{\stateswitch}[1]{\langle #1, \modeswitch\rangle}
\newcommand{\statepop}[1]{\langle #1, \modepop\rangle}

\newcommand{\inc}{\textsf{inc}}
\newcommand{\dec}{\textsf{dec}}
\newcommand{\eq}{a}
\newcommand{\bops}{\textsf{O}}
\newcommand{\val}{\textsf{val}}

\newcommand{\tikzfigBi}	{ 
	\begin{tikzpicture}[node distance=1.5cm, initial text=,>=latex, thick,]	
		
		\node[state,,,,inner sep=0pt,minimum size=12pt] (us0) {$s^i_0$};
		\node[state,right of=us0,inner sep=0pt,minimum size=12pt] (u1) {};
		\node[state,right of=u1,inner sep=0pt,minimum size=12pt] (u2) {};
		\node[rectangle, rounded corners, node distance = 0cm, right=of u2,  ] (ud1) {{$\cdots$}};
		\node[state,, right =0cm of ud1,inner sep=0pt,minimum size=12pt] (uo1) {};
		\node[state,,right of=uo1,inner sep=0pt,minimum size=12pt] (uo2) {};
		
		\node[state,above right of = uo2,minimum size=12pt] (uu1) {};
		\node[state, right  of = uu1,inner sep=0pt,,minimum size=12pt] (uu2) {};
		\node[state, right  of = uu2,inner sep=0pt,,minimum size=12pt] (uu3) {};
		\node[rectangle, rounded corners, node distance = 0cm, right=of uu3,  ] (ud2) {{$\cdots$}};
		\node[state,,, right = 0cm of ud2,inner sep=0pt,,minimum size=12pt] (ut1) {$t^i_0$};
		
		\node[state,right    of = uo2, node distance = 2.56cm ,minimum size=12pt] (um1) {};
		\node[state, right  of  = um1,inner sep=0pt,,minimum size=12pt] (um2) {};
		\node[rectangle, rounded corners, node distance = 0cm, right=of um2,  ] (ud3) {{$\cdots$}};
		\node[state,,, right = 0cm of ud3,inner sep=0pt,,minimum size=12pt] (ut2) {$t^i_1$};

		
		\node[state,below right  of = uo2,minimum size=12pt] (ub1) {};
		\node[state, right  of = ub1,inner sep=0pt,,minimum size=12pt] (ub2) {};
		\node[state, right  of =ub2,inner sep=0pt,,minimum size=12pt] (ub3) {};
		\node[rectangle, rounded corners, node distance = 0cm, right=of ub3,  ] (ud4) {{$\cdots$}};
		\node[state,,, right = 0cm of ud4,inner sep=0pt,,minimum size=12pt] (ut3) {$t^i_2$};
		
		
		\node[state,,below  = 3cm of us0,,inner sep=0pt,minimum size=12pt] (ds0) {$s^i_1$};
		\node[state,right of=ds0,inner sep=0pt,minimum size=12pt] (d1) {};
		\node[state,right of=d1,inner sep=0pt,minimum size=12pt] (d2) {};
		\node[rectangle, rounded corners, node distance = 0cm, right=of d2,  ] (dd1) {{$\cdots$}};
		\node[state,, right =0cm of dd1,inner sep=0pt,minimum size=12pt] (do1) {};
		\node[state,,right of=do1,inner sep=0pt,minimum size=12pt] (do2) {};
		
		\node[state,above right of = do2,minimum size=12pt] (du1) {};
		\node[state, right  of = du1,inner sep=0pt,,minimum size=12pt] (du2) {};
		\node[state, right  of = du2,inner sep=0pt,,minimum size=12pt] (du3) {};
		\node[rectangle, rounded corners, node distance = 0cm, right=of du3,  ] (dd2) {{$\cdots$}};
		\node[state,,, right = 0cm of dd2,inner sep=0pt,,minimum size=12pt] (dt1) {$t^i_3$};
		
		\node[state,right    of = do2, node distance = 2.56cm ,minimum size=12pt] (dm1) {};
		\node[state, right  of  = dm1,inner sep=0pt,,minimum size=12pt] (dm2) {};
		\node[rectangle, rounded corners, node distance = 0cm, right=of dm2,  ] (dd3) {{$\cdots$}};
		\node[state,,, right = 0cm of dd3,inner sep=0pt,,minimum size=12pt] (dt2) {$t^i_4$};

		
		\node[state,below right  of = do2,minimum size=12pt] (db1) {};
		\node[state, right  of = db1,inner sep=0pt,,minimum size=12pt] (db2) {};
		\node[state, right  of =db2,inner sep=0pt,,minimum size=12pt] (db3) {};
		\node[rectangle, rounded corners, node distance = 0cm, right=of db3,  ] (dd4) {{$\cdots$}};
		\node[state,,, right = 0cm of dd4,inner sep=0pt,,minimum size=12pt] (dt3) {$t^i_5$};

		\draw[->, blue ,inner sep=1.5pt] 
		(us0) edge node[above] { $\bzero,\bone$} (u1)
		(u1) edge node[above] { $\bzero,\bone$} (u2)
		(uo1) edge node[above] { $\bzero$} (uo2)
		(uo2) edge[bend left] node[above, sloped ] { $\bzero$} (uu1)
		(uu1) edge node[above,  ] { $\bzero$} (uu2)
		(uu2) edge node[above,  ] { $\bzero$} (uu3)
		(uu1) edge node[above,  sloped] { $\bone$} (um1)
		(uu2) edge node[above, sloped ] { $\bone$} (um2)
		(um1) edge node[above, at start, xshift=3mm  ] {$\bzero,\bone$} (um2)
		(uo2) edge[bend right] node[below, sloped ] { $\bone$} (ub1)
		(ub1) edge node[below,  ] { $\bone$} (ub2)
		(ub2) edge node[below,  ] { $\bone$} (ub3)
		(ub1) edge node[above, yshift=-0mm, xshift=-2mm, sloped ] { $\bzero$} (um1)
		(ub2) edge node[above, xshift=-2mm,sloped ] { $\bzero$} (um2)
		
		;
		\draw[->, blue, inner sep=1.5pt,] 
		(ds0) edge node[above] { $\bzero,\bone$} (d1)
		(d1) edge node[above] { $\bzero,\bone$} (d2)
		(do1) edge node[above] { $\bone$} (do2)
		(do2) edge[bend left] node[above, sloped ] { $\bzero$} (du1)
		(du1) edge node[above,  ] { $\bzero$} (du2)
		(du2) edge node[above,  ] { $\bzero$} (du3)
		(du1) edge node[above,  sloped] { $\bone$} (dm1)
		(du2) edge node[above, sloped ] { $\bone$} (dm2)
		(dm1) edge node[above, at start, xshift=3mm  ] {$\bzero,\bone$} (dm2)
		(do2) edge[bend right] node[below, sloped ] { $\bone$} (db1)
		(db1) edge node[below,  ] { $\bone$} (db2)
		(db2) edge node[below,  ] { $\bone$} (db3)
		(db1) edge node[above, yshift=-0mm, xshift=-2mm, sloped ] { $\bzero$} (dm1)
		(db2) edge node[above, xshift=-2mm,sloped ] { $\bzero$} (dm2)
		
		;

		
		\tikzset{node distance=1cm, fill=black, ,>=latex, thick}			
		\node[fill = black, , circle, inner sep = 0pt, minimum size = 0pt, outer sep = 0pt,  above  right of = ut1, yshift = -0cm] (ur) {};	
		
		\node[fill = black, , circle, inner sep = 0pt, minimum size = 0pt, outer sep = 0pt,  below right  of = dt3] (lr) {};
		
		\draw[->, brown] (ur) .. controls +(180:6) and +(70:2) .. node[near start, above ] {$\Gamma_2$} (us0);
		\draw[brown] (ut3) -| (ur);
		\draw[brown] (ut2) -| (ur); 
		\draw[brown ] (ut1) -| (ur);
		
		\draw[->, brown ] (lr) .. controls +(180:6) and +(290:2) .. node[near start, below  ] {$\Gamma_2$}  (ds0);
		\draw[brown] (dt1) -| (lr) ; 	
		\draw[brown] (dt2) -| (lr);
		\draw[brown] (dt3) -| (lr); 
		
		
		\tikzset{node distance=2cm, fill=black, ,>=latex, thick}			
		\node[fill = black, , circle, inner sep = 0pt, minimum size = 0pt, outer sep = 0pt,  above  right of = ut1, xshift=+3mm] (ur) {};	
		
		\node[fill = black, , circle, inner sep = 0pt, minimum size = 0pt, outer sep = 0pt,  below right  of = dt3, ] (lr) {};
		
		\draw[->, green!50!black] (ur) .. controls +(180:6) and +(90:2) .. node[above] {$\inc$} (us0);
		\draw[green!50!black] (dt3)+(0.2, 0.1) -| (ur);
		\draw[green!50!black] (ut2)+(0.2, 0.1) -| (ur); 
		\draw[green!50!black ] (ut1)+(0.2, 0.1) -| (ur);
		
		\draw[->, green!50!black ] (lr) .. controls +(180:6) and +(270:2) .. node[below] {$\inc$}  (ds0);
		\draw[green!50!black] (dt1)+(0.2, 0.1) -| (lr) ; 	
		\draw[green!50!black] (dt2)+(0.2, 0.1) -| (lr);
		\draw[green!50!black] (ut3)+(0.2, 0.1) -| (lr); 
		
		
		\tikzset{node distance=4cm, fill=black, ,>=latex, thick}			
		\node[fill = black, , circle, inner sep = 0pt, minimum size = 0pt, outer sep = 0pt,  above  right of = ut1, yshift = -0.5cm, ] (ur) {};	
		
		\node[fill = black, , circle, inner sep = 0pt, minimum size = 0pt, outer sep = 0pt,  below right  of = dt3, yshift = 0.5cm, xshift=+3mm] (lr) {};
		
		\draw[->, red] (ur) .. controls +(180:6) and +(110:2.5) .. node[near end, above] {$\dec$} (us0);
		\draw[red] (ut3)+(0.2, -0.1) -| (ur);
		\draw[red] (ut2)+(0.2, -0.1) -| (ur); 
		\draw[red] (dt1)+(0.2, -0.1) -| (ur);
		
		\draw[->, red] (lr) .. controls +(180:6) and +(250:2.5) .. node[near end, below] {$\dec$}  (ds0);
		\draw[red] (ut1)+(0.2, -0.1) -| (lr) ; 	
		\draw[red] (dt2)+(0.2, -0.1) -| (lr);
		\draw[red] (dt3)+(0.2, -0.1) -| (lr); 
		
	\end{tikzpicture}
}

\newcommand{\tikzfigAone}{\begin{tikzpicture}[node distance=2.2cm, fill=blue!10, initial text=,>=latex, thick]
	
	\node[state,fill=blue!10,,inner sep=0pt,minimum size=12pt] (us0) {$s^1_0$};

	\node[state,fill=blue!10,right = 1cm of us0,inner sep=0pt,minimum size=12pt] (u1) {};
	
	\node[state,fill=blue!10,,above right  of = u1,,inner sep=0pt,minimum size=12pt] (ut1) {$t^1_0$};

	\node[state,fill=blue!10,right of =  u1, node distance = 1.51cm ,,inner sep=0pt, ,minimum size=12pt] (ut2) {$t^1_1$};			
	\node[state,fill=blue!10,below right of =  u1,,inner sep=0pt,minimum size=12pt] (ut3) {$t^1_2$};
	\draw[->, inner sep=1pt] 
	(us0) edge node[above] {$ \bzero$} (u1)						
	(u1) edge[dotted, bend left] node[above,sloped] { $\bzero^{\ell-1}$} (ut1)
	(u1) edge[dotted, bend right] node[below,sloped] { $\bone^{\ell-1}$} (ut3)		
	;
	
	\node[state,fill=blue!10,below of=us0,node distance = 5cm,inner sep=0pt,minimum size=12pt] (ds0) {$s^1_1$};
	
	\node[state,fill=blue!10,right = 1cm of ds0,inner sep=0pt,minimum size=12pt] (d1) {};
	
	\node[state,fill=blue!10,,above right  of = d1,,inner sep=0pt,minimum size=12pt] (dt1) {$t^i_3$};
	\node[state,fill=blue!10,right of =  d1, node distance = 1.51cm ,,inner sep=0pt, ,minimum size=12pt] (dt2) {$t^1_4$};			
	\node[state,fill=blue!10,below right of =  d1,,inner sep=0pt,minimum size=12pt] (dt3) {$t^1_5$};
	
	\tikzset{node distance=1cm, fill=black, ,>=latex, thick}
	
	\node[draw = none, circle, minimum size=0pt,inner sep=0pt, outer sep = 0pt, right of = ut1] (ut11) {}; 
	\node[fill = black, , circle, inner sep = 0pt, minimum size = 0pt, outer sep = 0pt,  above  of = ut11, yshift = -0cm] (ur) {};
	
	\tikzset{node distance=1cm, fill=black, ,>=latex, thick}
	
	\node[draw = none, circle, minimum size=0pt,inner sep=0pt, outer sep = 0pt, right of = dt3] (dt33) {}; 
	\node[fill = black, , circle, inner sep = 0pt, minimum size = 0pt, outer sep = 0pt,  below  of = dt33] (lr) {};

	\draw[->] (ur) edge[bend right = 20] node[above] {$\Gamma_2$} (us0);
	\draw (ut3) -| (ur);
	\draw (ut2) -| (ur); 
	\draw (ut1) -| (ur);
	
	\draw[->] (lr) edge[bend left = 20] node[below] {$\Gamma_2$} (ds0);
	\draw (dt1) -| (lr) ; 	
	\draw (dt2) -| (lr);
	\draw (dt3) -| (lr); 
	
	\tikzset{node distance=1.5cm, fill=black, ,>=latex, thick}
	
	\node[draw = none, circle, minimum size=0pt,inner sep=0pt, outer sep = 0pt, right of = ut1] (ut11) {}; 
	\node[fill = black, , circle, inner sep = 0pt, minimum size = 0pt, outer sep = 0pt,  above  of = ut11, yshift = -0cm] (ur) {};
	
	\tikzset{node distance=1.5cm, fill=black, ,>=latex, thick, }
	
	\node[draw = none, circle, minimum size=0pt,inner sep=0pt, outer sep = 0pt, right of = dt3] (dt33) {}; 
	\node[fill = black, , circle, inner sep = 0pt, minimum size = 0pt, outer sep = 0pt,  below  of = dt33] (lr) {};

	\draw[->, green!70!black!70] (ur) edge[bend right = 40 ] node[above] {$\inc$} (us0);
	\draw[green!70!black!70] (ut2)+(0.2, 0.1) -| (ur); 
	\draw[green!70!black!70] (ut1)+(0.2,0.1) -| (ur);
	
	\draw[->, green!70!black!70] (lr) edge[bend left = 40]  node[below] {$\inc$} (ds0);
	\draw[green!70!black!70,] (ut3)+(0.2,0.1) -| (lr);
	\draw[green!70!black!70] (dt1)+(0.2,0.1) -| (lr) ; 	
	\draw[green!70!black!70] (dt2)+(0.2,0.1) -| (lr);
	
	\tikzset{node distance=2cm, fill=black, ,>=latex, thick}
	
	\node[draw = none, circle, minimum size=0pt,inner sep=0pt, outer sep = 0pt, right of = ut1] (ut11) {}; 
	\node[fill = black, , circle, inner sep = 0pt, minimum size = 0pt, outer sep = 0pt,  above  of = ut11, yshift = -0cm] (ur) {};

	\tikzset{node distance=2cm, fill=black, ,>=latex, thick, }
	
	\node[draw = none, circle, minimum size=0pt,inner sep=0pt, outer sep = 0pt, right of = dt3] (dt33) {}; 
	\node[fill = black, , circle, inner sep = 0pt, minimum size = 0pt, outer sep = 0pt,  below  of = dt33] (lr) {};

	\draw[->, red] (ur) edge[bend right = 60 ] node[above] {$\dec$} (us0);
	\draw[red] (ut2)+(0.2, -0.1) -| (ur); 
	\draw[red] (ut3)+(0.2,-0.1) -| (ur);
	\draw[red] (dt1)+(0.2,-0.1) -| (ur) ; 	
	
	\draw[->, red] (lr) edge[bend left = 60] node[below] {$\dec$} (ds0);
	\draw[red,] (dt3)+(0.2,-0.1) -| (lr);
	
	\draw[red] (dt2)+(0.2,-0.1) -| (lr);

	\draw[->, inner sep=1pt] 
	(ds0) edge node[above] {$ \bone$} (d1)						
	(d1) edge[dotted, bend left] node[above,sloped] { $\bzero^{\ell-1}$} (dt1)
	(d1) edge[dotted, bend right] node[below,sloped] { $\bone^{\ell-1}$} (dt3)
	;	
\end{tikzpicture}}

Let $\Gamma_1 = \{\bzero, \bone, \inc, \dec\}$, and let $\Gamma_2$ be another finite alphabet disjoint from $\Gamma_1$.   Our objective is to come up with a PDA $A$ and $\ell$ DFAs  $A_1, A_2, \dots A_\ell$ over the alphabet $\Gamma_1 \cup \Gamma_2$, each of size $\mathcal{O}(\ell)$ such that any word accepted by all of them (i.e., in  $ \cap_i L(A_i) \cap L(A)$) has $2^{2^\ell}$ occurrences of letters from $\Gamma_2$.

\medskip

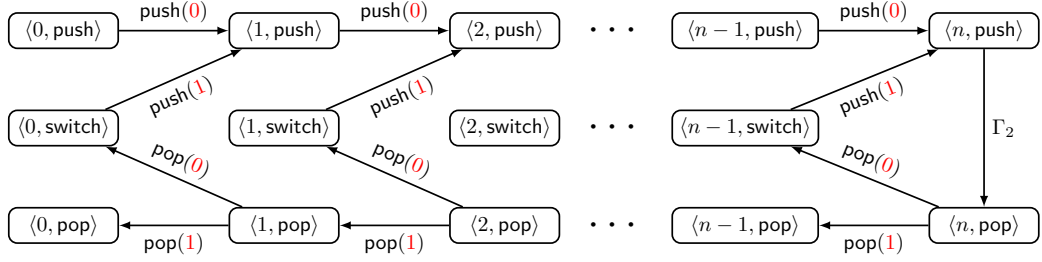
\begin{figure}
	\scalebox{0.8}{
		\begin{tikzpicture}[node distance=1.8cm, initial text=,>=latex, thick, minimum width=1.8cm]
		\node[rectangle, draw, rounded corners, ] (t0) {$\statepush{0}$};
		\node[rectangle, draw, rounded corners, right=of t0 ] (t1) {$\statepush{1}$};
		\node[rectangle, draw, rounded corners, right=of t1 ] (t2) {$\statepush{2}$};	
		\node[rectangle, rounded corners, node distance = 0cm, right=of t2,  ] (t3) {\scalebox{2}{$\cdots$}};
			\node[rectangle, draw, rounded corners,node distance = 0cm,  right=of t3 , minimum width=2.4cm ] (t4) {$\statepush{n-1}$};	
					\node[rectangle, draw, rounded corners, right=of t4 ] (t5) {$\statepush{n}$};	
		\node[rectangle, draw, rounded corners, node distance = 1cm,  below=of t0 ] (m0) {$\stateswitch{0}$};
\node[rectangle, draw, rounded corners, right=of m0 ] (m1) {$\stateswitch{1}$};
\node[rectangle, draw, rounded corners, right=of m1 ] (m2) {$\stateswitch{2}$};	
\node[rectangle, rounded corners, node distance = 0cm, right=of m2, ] (m3) {\scalebox{2}{$\cdots$}};
\node[rectangle, draw, rounded corners, node distance = 0cm, right=of m3, minimum width=2.4cm ] (m4) {$\stateswitch{n-1}$};
		\node[rectangle, draw, rounded corners, node distance = 1cm,  below=of m0,  ] (b0) {$\statepop{0}$};
\node[rectangle, draw, rounded corners, right=of b0 ] (b1) {$\statepop{1}$};
\node[rectangle, draw, rounded corners, right=of b1 ] (b2) {$\statepop{2}$};	
\node[rectangle, rounded corners, node distance = 0cm, right=of b2, ] (b3) {\scalebox{2}{$\cdots$}};
\node[rectangle, draw, rounded corners, node distance = 0cm, right=of b3, minimum width=2.4cm  ] (b4) {$\statepop{n-1}$};	
					\node[rectangle, draw, rounded corners, right=of b4 ] (b5) {$\statepop{n}$};	
	\draw[->] 
(t0) edge node[above] {$\push(\rzero)$} (t1)
(t1) edge node[above] {$\push(\rzero)$} (t2)
(t4) edge node[above] {$\push(\rzero)$} (t5)
(m0) edge[] node[below, sloped] {$\push(\rone)$} (t1)
(m4) edge node[below, sloped] {$\push(\rone)$} (t5)
(m1) edge node[below, sloped] {$\push(\rone)$} (t2)
(b2) edge node[above, sloped] {$\pop(\rzero)$} (m1)
(b1) edge node[above, sloped] {$\pop(\rzero)$} (m0)
(b5) edge node[below] {$\pop(\rone)$} (b4)
(b2) edge node[below] {$\pop(\rone)$} (b1)
(b1) edge node[below] {$\pop(\rone)$} (b0)
(b5) edge node[above, sloped] {$\pop(\rzero)$} (m4)
(t5) edge[] node[ minimum width = 0cm, right] {$\Gamma_2$} (b5)
;
	\end{tikzpicture}
}
	\caption{A PDA with $3n+2$ states that accepts  $(\Gamma_2)^{2^n}$.}\label{fig:exp-pda}
\end{figure}

First we give a PDA with $3n+2$ states and  stack symbols $\{\rzero, \rone\}$ that accepts $(\Gamma_2)^{2^n}$.
\begin{claim}
	There is a PDA with $3n+2$ states and  stack symbols $\{\rzero, \rone\}$ with stack-height never exceeding $n$ that accepts $(\Gamma_2)^{2^n}$.
\end{claim}
\begin{proof}
	Such a PDA is depicted in Figure~\ref{fig:exp-pda}.  In this PDA the stack height never exceeds $n$. The PDA has three modes - a $\modepush$ mode where it keeps pushing $\rzero$s until the stack height is $n$, a $\modepop$ mode where it keeps popping the symbol $\rone$, and a $\modeswitch$ mode that switches from a $\modepop$ mode to $\modepush$  mode by replacing a $\rzero$ on the top of the stack  by a $\rone$.  The states then represent the current stack height and the mode. 
	
	When the PDA is in the state $\statepush{n}$, the stack contents represents an $n$-bit binary number. At this point it reads a symbol from $\Gamma_2$ and goes to the state $\statepop{n}$. From there, it does a sequence of transitions such that the next time it reaches $\statepush{n}$, the binary number in the stack would be incremented. For this it replaces the $\rzero\rone^m$ suffix with a $\rone\rzero^m$.
	
	The initial state is $\statepush 0$. The first time it reaches $\statepush n$, the stack content would be $\rzero^n$. The last time it reaches $\statepush n$ (or $\statepop n$) the stack content would be $\rone^n$, and from there it reaches $\statepop 0$ by popping the entire stack. The state $\statepop 0$ is accepting. Since this PDA reaches the state $\statepush{n}$ exactly once for every $n$-bit binary number, the number of $\Gamma_2$ symbols it reads is $2^n$. 
\end{proof}

\begin{figure}
	\scalebox{0.8}{
	\begin{tikzpicture}[node distance=1.8cm, initial text=,>=latex, thick, minimum width=1.8cm]
		\node[rectangle, draw, rounded corners, ] (t0) {$\statepush{0}$};
		\node[rectangle, draw, rounded corners, right=of t0 ] (t1) {$\statepush{1}$};
		\node[rectangle, draw, rounded corners, right=of t1 ] (t2) {$\statepush{2}$};	
		\node[rectangle, rounded corners, node distance = 0cm, right=of t2,  ] (t3) {\scalebox{2}{$\cdots$}};
		\node[rectangle, draw, rounded corners,node distance = 0cm,  right=of t3 , minimum width=2.4cm ] (t4) {$\statepush{n-1}$};	
		\node[rectangle, draw, rounded corners, right=of t4 ] (t5) {$\statepush{n}$};	
		\node[rectangle, draw, rounded corners, node distance = 1cm,  below=of t0 ] (m0) {$\stateswitch{0}$};
		\node[rectangle, draw, rounded corners, right=of m0 ] (m1) {$\stateswitch{1}$};
		\node[rectangle, draw, rounded corners, right=of m1 ] (m2) {$\stateswitch{2}$};	
		\node[rectangle, rounded corners, node distance = 0cm, right=of m2, ] (m3) {\scalebox{2}{$\cdots$}};
		\node[rectangle, draw, rounded corners, node distance = 0cm, right=of m3, minimum width=2.4cm ] (m4) {$\stateswitch{n-1}$};
						\node[rectangle, draw, rounded corners, right=of m4 ] (m5) {$\stateswitch{n}$};	
		\node[rectangle, draw, rounded corners, node distance = 1cm,  below=of m0,  ] (b0) {$\statepop{0}$};
		\node[rectangle, draw, rounded corners, right=of b0 ] (b1) {$\statepop{1}$};
		\node[rectangle, draw, rounded corners, right=of b1 ] (b2) {$\statepop{2}$};	
		\node[rectangle, rounded corners, node distance = 0cm, right=of b2, ] (b3) {\scalebox{2}{$\cdots$}};
		\node[rectangle, draw, rounded corners, node distance = 0cm, right=of b3, minimum width=2.4cm  ] (b4) {$\statepop{n-1}$};	
		\node[rectangle, draw, rounded corners, right=of b4 ] (b5) {$\statepop{n}$};	
		\draw[->] 
		(t0) edge node[above] {$\push(\rzero)$} (t1)
		(t1) edge node[above] {$\push(\rzero)$} (t2)
		(t4) edge node[above] {$\push(\rzero)$} (t5)
		(m0) edge[] node[below, sloped] {$\push(\rone)$} (t1)
		(m4) edge node[below, sloped] {$\push(\rone)$} (t5)
		(m1) edge node[below, sloped] {$\push(\rone)$} (t2)
		(b2) edge node[above, sloped] {$\pop(\rzero)$} (m1)
		(b1) edge node[above, sloped] {$\pop(\rzero)$} (m0)
		(b5) edge node[below] {$\pop(\rone)$} (b4)
		(b2) edge node[below] {$\pop(\rone)$} (b1)
		(b1) edge node[below] {$\pop(\rone)$} (b0)
		(b5) edge node[above, sloped] {$\pop(\rzero)$} (m4)
		(t5) edge[] node[ minimum width = 0cm, right] {$\Gamma_2$} (m5)
				(m5) edge[] node[ minimum width = 0cm, right] {$\Gamma_2$} (b5)
		;
	\end{tikzpicture}
}
	\caption{A PDA with $3n+3$ states that accepts  $(\Gamma_2)^{2^{n+1}}$.}\label{fig:exp-pda-1}
\end{figure}
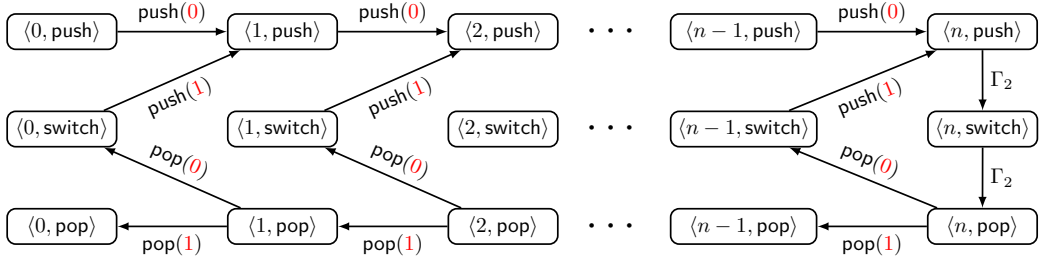

If $n = 2^\ell$ we will be able to get $2^{2^\ell}$ length words from $\Gamma_2$ as we wanted. However, we are allowed to use only $\mathcal{O}(\ell)$ states. To overcome this, we will use $\ell$ length binary numbers to indicate the current stack height. We then use $\ell$ DFAs, one for each bit, to update the binary numbers representing the stack height as required. We describe this below. 

There is a small caveat that we need to address first. Since we will be using $\ell$ bit numbers for representing the stack height, the maximum height we can faithfully represent is $2^\ell -1$. With this, we would only get $ (\Gamma_2)^{2^{\ell}-1}$ and not the desired $(\Gamma_2)^{2^{\ell}}$. Hence we modify the PDA in Figure~\ref{fig:exp-pda} to accept words of length $2^{n+1}$ by adding an extra state.  This new PDA  is depicted  in Figure~\ref{fig:exp-pda-1}. It reads two $\Gamma_2$ symbols each time a different $n$-bit number is present in the stack.  Thus the language of this $3n+3$ state PDA is $(\Gamma_2)^{2^{n+1}}$. 
\begin{claim}
	There is a PDA with $3n+3$ states and  stack symbols $\{\rzero, \rone\}$ with stack-height never exceeding $n$ that accepts $(\Gamma_2)^{2^{n+1}}$.
\end{claim}

We will next describe the $\ell$ DFAs that succinctly record the stack height. We then give a PDA that, along with these $\ell$ DFAs,  accepts words with $2^{2^\ell}$ occurrences of letters from $\Gamma_2$. 

\medskip

Let $\{\inc, \dec\}$ disjoint from $\Gamma_2$ be the increment and the decrement operators on integers. Further, we may treat symbols from $\Gamma_2$ as `keep unchanged' operators. That is, $\inc(n) = n+1$, $\dec(n) = n-1$ and $\eq(n) = n$ for all $\eq \in \Gamma_2$.   Consider the following language over the alphabet $\Gamma = \{\bzero, \bone, \inc, \dec\} \cup \Gamma_2$ of alternating sequences of $\ell$-bit numbers and operators, where each operator when applied on the previous number gives the next number. Here, the binary numbers are written with the most-significant bit on the left. That is $\val(b_{1}b_{2} \dots b_\ell) = \sum_i b_i\times 2^{\ell-i}$. 
\begin{eqnarray*}
	L_\ell = \{n_0o_1n_1o_2n_2 \cdots o_kn_k &\mid&k \ge 0, n_i \in (\bzero+\bone)^\ell \textrm { for all }i: 0 \le i \le k \\
	&&	o_i \in \{\inc, \dec\} \cup \Gamma_2 \textrm{  for all }i: 0 < i \le k\\
	&& \val(n_{i}) \equiv o_{i} (\val(n_{i-1}))\mod 2^\ell \textrm{ for all }i:  0 < i \le k\}
\end{eqnarray*}

\begin{claim}
	There are  $\ell$ DFAs $B_1, B_2, \dots B_\ell$, each with $\mathcal{O}(\ell)$ states such that $L_\ell = \bigcap_i L(B_i)$. 
\end{claim}

\begin{proof}
	We describe  the $\ell$ DFAs $B_1, B_2, \dots B_\ell$  below. 
	
	The $i$th DFA $ B_i$ guarantees that the $i$th bit takes the correct value. This DFA is depicted in the Figure~\ref{fig:Aut}.
	The automaton has two disconnected `forks' (the top one starting at $s^i_0 $ and the bottom one starting at $s^i_1 $). In the top fork, the $i$th bit read is always $\bzero$, and in the bottom fork the $i$th bit read is always $\bone$. 
	Consider $n_{j-1} o_{j} n_{j}$ occurring in the above sequence. Let $n_{j-1} = b_{1}b_{2} \dots b_\ell$ and let $n_{j } = b'_{1}b'_{2} \dots b'_\ell$. If $o_{j}$ is $\inc$, the $i$th bit $b_i$ is toggled ($b'_i \neq b_i$) iff $b_{m} = 1$ for all $m: m > i$. If $o_{j}$ is $\dec$, the $i$th bit $b_i$ is toggled ($b'_i \neq b_i$) iff $b_{m} = 0$ for all $m: m > i$. If $o_{j} \in \Gamma_2$ the $i$th bit is never toggled. The initial states are  $s^i_0$ and $s^i_1$, and the accepting states are $t^i_0, \dots t^i_5$. Clearly $L_\ell = \bigcap_i L(B_i)$. 
\end{proof}

\begin{figure}
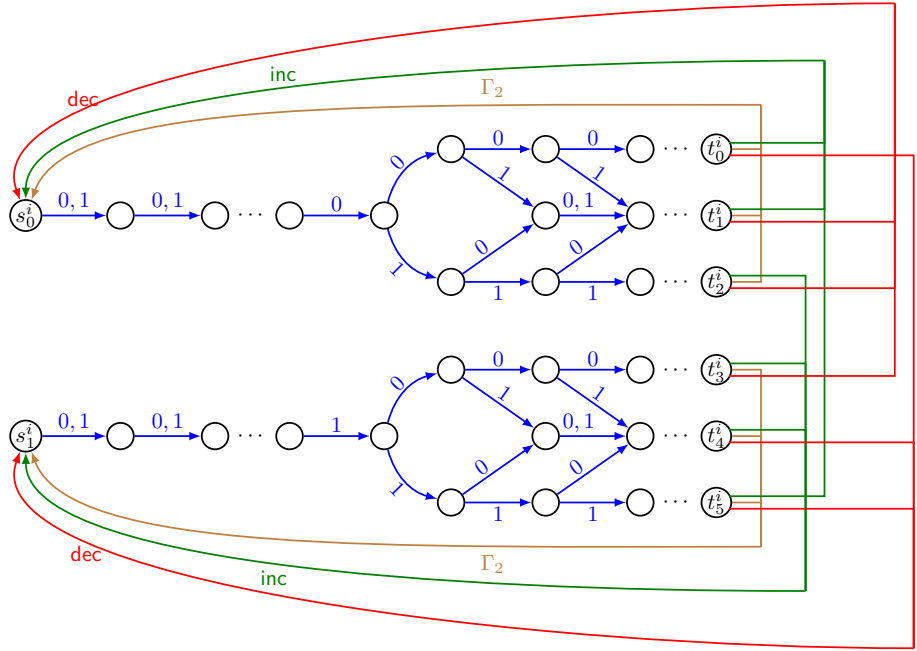

	\centering
\scalebox{.83}{\tikzfigBi}
	\caption{The automaton $B_i$. The transitions on $\{\bzero, \bone\} $ are depicted in blue. The transitions on $\Gamma_2$ (resp. $\inc$, $\dec$) are depicted in brown (resp. green, red).   } \label{fig:Aut}
\end{figure}

	However, for succinctly simulating the PDA given in Figure~\ref{fig:exp-pda-1},  we need the $\ell$ DFAs to faithfully reflect the stack height. For this we consider a slight modification of $L_\ell$. 
	\begin{eqnarray*}
		L'_\ell = \{ n_0o_1n_1o_2n_2 \cdots o_kn_k &\mid & k \ge 0, n_i \in (\bzero+\bone)^\ell \textrm { for all }i: 0 \le i \le k \\
		&&	o_i \in \{\inc, \dec\} \cup \Gamma_2 \textrm{  for all }i: 0 < i \le k\\ &&\val(n_{i}) \equiv o_{i} (\val(n_{i-1})) \textrm{ for all } i:  0 < i \le k\\
		&& n_0 = \bzero^\ell = n_k, \textrm{ if } o_i = \inc \textrm{ then } n_{i-1} \neq \bone^\ell  \\
		&& \textrm{ if } o_i = \dec \textrm{ then } n_{i-1} \neq \bzero^\ell\}
	\end{eqnarray*}
	This ensures that the PDA starts and ends with an empty stack. Further $\inc$ after $1^\ell$ and $\dec$ after $0^\ell$ are forbidden. Otherwise, the value will not  faithfully represent the stack height.

	\begin{claim}\label{claim:countingDFA}
		There are  $\ell$ DFAs $A_1, A_2, \dots A_\ell$, each with $\mathcal{O}(\ell)$ states such that $L'_\ell = \bigcap_i L(A_i)$. 
	\end{claim}
\begin{proof}
	The states of $A_i$ are exactly those of $B_i$. For $i\ge 2$, the transitions of $A_i$ is exactly the same as that of $B_i$. The transitions for $A_1$ is obtained by removing two transitions from that of $B_1$, namely the outgoing $\inc$ transition from $t^1_5$ and the outgoing $\dec$ transition from $t^1_0$. For all $i\ge 1$, the initial state of $A_i$ is $s^i_0$ and the final state is $t^i_0$. 
\end{proof}

\begin{claim}\label{claim:countingPDA}
	There is a PDA  $A$ with 3 states and stack symbols $\{\rzero, \rone\}$ such that $L(A) \cap L'_\ell$ when projected to $\Gamma_2$ is exactly $(\Gamma_2)^{2^{2^\ell}}$.
\end{claim}
\noindent
\begin{minipage}{0.4\textwidth}
	\scalebox{1}{
	\begin{tikzpicture}
		\node[rectangle, draw, rounded corners, ] (t0) {$\modepush$};
		\node[rectangle, draw, rounded corners, below=of t0 ] (t1) {$\modeswitch$};
		\node[rectangle, draw, rounded corners, below=of t1 ] (t2) {$\modepop$};	
		
		\draw[->] 
		(t1) edge[bend left] node[left] { $\inc\!\mid\!\push(\rone)$} (t0)
		(t2) edge [bend left]node[left] { $\dec\!\mid\!\pop(\rzero)$} (t1)
		(t0) edge[bend left] node[right] { $\Gamma_2$} (t1)
		(t1) edge [bend left]node[right] { $\Gamma_2$} (t2)
		(t0) edge[loop right] node[right] { $\inc\!\mid\!\push(\rzero)$} (t0)
		(t0) edge[loop left] node[left] { $\bone,\bzero$} (t0)
		(t1) edge[loop left] node[left] { $\bone,\bzero$} (t1)
		(t2) edge[loop right] node[right] { $\dec\!\mid\! \pop(\rone)$} (t2)
		(t2) edge[loop left] node[left] { $\bone,\bzero$} (t2)
		
		;
	\end{tikzpicture}}
	
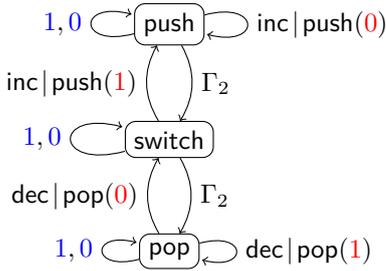
\captionof{figure}{The PDA  $A$. This PDA and the $\ell$ DFAs together faithfully encode the accepting runs of the PDA in Figure~\ref{fig:exp-pda-1} with $n = 2^\ell-1$.  Thus they accept words with exactly $2^{2^\ell}$ occurrences of $\Gamma_2$.
 }\label{fig:dexpPDA}
\end{minipage}\hfill
\begin{minipage}{0.55\textwidth}
\begin{proof}
	The PDA $A$ is depicted in Figure~\ref{fig:dexpPDA}. The three states represents the three modes of the PDA in Figure~\ref{fig:exp-pda-1}.
	The $\ell$ DFAs will guarantee that we start with the number $\bzero^\ell$. Because of $A_1$, the PDA cannot take the $\inc$ transition from the state $\langle \modepush \rangle$ immediately after the number  $\bone^\ell$. It will have to read a $\Gamma_2$ symbol and move to the state $\langle \modeswitch \rangle $. The  PDA will loop in this state once by reading the same number ($\bone^\ell$) as mandated by  the $\Gamma_2$ transitions of the DFAs.  From this state, again $\inc$ is disabled by $A_1$, and hence the PDA will read another $\Gamma_2$ symbol and go to the state $\langle \modepop \rangle$. The PDA will read $\bone^\ell$ staying in the state $\langle \modepop \rangle $ after which it can take a $\dec$ transition. 
\end{proof}
\end{minipage}

\newcommand{\pda}{P}
\subsection{Completing the reduction}
\newcommand{\pcp}{\mathcal{P}}
\def\dsqcup{\sqcup\mathchoice{\mkern-7mu}{\mkern-7mu}{\mkern-3.2mu}{\mkern-3.8mu}\sqcup}
\newcommand{\shsymbol}{\mathfrak{W}}
Before giving the reduction, let us first define a PDA and the transducers that we use in the reduction. Let $\pda_1$ be the PDA in Claim~\ref{claim:countingPDA} with $\Gamma_2 = \Sigma_1 \cup \{\#\}$. Let $\pda_2 $ the PDA for $L(\pda_1)\cap (\Sigma_1^\ast \#^\ast)$. Let $\pda_3$ be the PDA in Claim~\ref{claim:countingPDA} with $\Gamma_2 = \Sigma_2$. Let $\pda_4$ be the PDA for $L(\pda_2)\cdot L(\pda_3)$.  
Let  $T_1, \dots T_\ell$ be $\ell$ transducers. The input automaton of $T_i$ is exactly the DFA  $A_i$ from Claim~\ref{claim:countingDFA} with $\Gamma_2 = \Sigma_1 \cup \Sigma_2 \cup \{\#\}$. The output of the transducer  $T_i$ on every transition is $\epsilon$. 

\medskip

Now we are ready to give the reduction. Let $\pcp = (\Sigma_1, \Sigma_2, f, g, \ell)$ be an input to a 2eBPCP problem. We describe how to obtain a string constraint $C_P$ from this. 
We use the alphabet $ \Sigma = \Sigma_1 \cup \Sigma_2 \cup \Gamma_1 \cup \{\#\}$, and the variable set $\variableset = \{x_0, x_1, x_2, x_f, x_g\}$. Next we define the membership constraints $\Mem$.  Let $\Mem(x_0) = A_\epsilon$  where $A_\epsilon$ is an NFA for $\{\epsilon\}$. We have  $ \Mem(x_1) =\pda_4, \Mem(x_2) = \pda_4$.  Now we need to augment the language of $x_f$ and $x_g$ to also account for the letters from $\Gamma_1$, which can be achieved by adding $\Gamma_1$ self loops in all the states in the PDA for $L_f$ and  $L_g$ respectively. Thus the language for $x_f$ is $\shsymbol(L_f, \Gamma_1^\ast)$, where $\shsymbol(L_f, \Gamma_1^\ast) = \{w \mid w \in \shuffle{u,v}, u \in L_f, v\in \Gamma_1^\ast\}$. Similarly the language for $x_g$ is $\shsymbol(L_g, \Gamma_1^\ast)$. Let $\pda_5 $ and $\pda_6$ be  PDAs recognizing $\shsymbol(L_f, \Gamma_1^\ast)$, and $\shsymbol(L_g, \Gamma_1^\ast)$ respectively. We have $\Mem(x_f) = \pda_5$ and $\Mem(x_g) = \pda_6$. 
Let $\Rel$ be the following relational constraints: 1)$ x_0 \subword T_i(x_1)$, for all $i: 1< i < \ell$, 2) $x_0 \subword T_i(x_2)$, for all $i: 1< i < \ell$, 3) $x_1 \subword x_f$, 4) $x_f \subword x_g$ and 5) $x_g \subword x_2$. We have $2\ell +3$ relational constraints. Further $\multiplicity(x_1) = \ell  = \multiplicity(x_2)$ in our construction.
Let $C_\pcp = (\Sigma, \variableset, \Mem, \Rel)$. 
\begin{claim}\label{claim:2eBPCP}
	The string constraint $C_\pcp$ is satisfiable if and only if the 2eBPCP instance $\pcp$ has a solution.
\end{claim}
\begin{proof}
	Suppose the string constraint $C_\pcp$ is satisfiable, let $\assignment$ be the satisfying assignment.  Then we claim that $\assignment(x_f) = \assignment(x_g) = \assignment(x_1) = \assignment(x_2)$.
	Further more, $\proj{\assignment(x_f)}{\Gamma_2} = \Sigma_1^{m_1} \#^{m_2} \Sigma_2^{m_1 + m_2}$, where $m_1$ and $m_2$ are such that $m_1 + m_2 = 2^{2^\ell}$. This, combining with the membership constraints of $x_f,x_g$ immediately provides us with the solution for the 2eBPCP  instance $\pcp$. 
	
	To see why $\assignment(x_f) = \assignment(x_g) = \assignment(x_1) = \assignment(x_2)$, note that the constraints in $1$ and $2$ will ensure that $ \assignment(x_1) \in \bigcap_{i = 1}^\ell L(A_i)$.
	The membership constraint $\Mem(x_1)= P_4  $ will entail $ \assignment(x_1) = w_1 \cdot w_2$ with $w_1 \in L(P_2)$ and $w_2 \in L(P_3)$. This will in turn ensure that $\proj{w_1}{\Gamma_2} = \Sigma_1^{m_1} \#^{m_2}$, where $m_1 + m_2 = 2^{2^\ell}$ and $\proj{w_1}{\Gamma_2}  = \Sigma_2^{2^{2^\ell}}$. Similar argument can also be made for $\sigma(x_2)$. From this, we also obtain that $|\sigma(x_1)| = |\sigma(x_2)|$. From the relational constraints $3,4,5$, we also have $x_1 \subword x_2$ implying $\sigma(x_1) = \sigma(x_2)$.
	
	\medskip
	For the other direction, suppose $\pcp$ has a solution, then we construct the satisfying assignment for $C_\pcp$ as follows. Let the solution for $\pcp$ be $w \in\Sigma_1^*$ such that $f(w) = g(w)$ and $|f(w)| = 2^{2^\ell}$. Let $w' \in L'_\ell$ be such that $\proj{w'}{\Sigma_1 \cup \Sigma_2 \cup \{\#\}} = v$, where $v = w \cdot \#^m \cdot f(w)^{r}$, for some $m = 2^{2^\ell} - |w|$. We let $\assignment(x_1) = \assignment(x_2) = \assignment(x_g) = \assignment(x_f) = w'$. Notice that the constraints $1,2$ are satisfied since the word is picked from $L'_\ell$. It satisfies the  constraints $3,4,5$ since all the variables are assigned the same word. It is also easy to verify that the membership constraints are satisfied.  Hence, $\sigma$ that we constructed is a satisfying assignment.
\end{proof}

Notice that we construct $C_\pcp$ from $\pcp$ in polynomial time, and it is acyclic.  Hence, it follows that the satisfiability checking of acyclic string constraints is \twonexph, proving the lower bound of Theorem~\ref{thm:2nexpcCFL}.

\section{Satisfiability is in 2NEXPTIME}\label{sec:2nexptimeUpperBound}
\newcommand{\constraint}[2]{\textsf{constraint}(#1,#2)}
\newcommand{\Transducer}[2]{\transducer_{#2}^{#1}}
\newcommand{\outputvar}[2]{o_{#2}^{#1}}
\newcommand{\extendedvariableset}{\widehat{\variableset}}
\newcommand{\extendedassignment}{\widehat{\assignment}}
\newcommand{\bound}{B}

We will show that if an acyclic string constraint is satisfiable, then there is a satisfying assignment of double exponential size. 

\begin{theorem}[Small model property]\label{thm:smallmodel} Let $C$ be a satisfiable acyclic string constraint. Then $C$ has a satisfying assignment $\assignment$ such that $\len{\assignment(x)}  \le \bound$ where $\bound$ is $$2^{2^{\mathcal{O}(m\log t+ \log p + \log k)}}$$ where 
	\begin{itemize}
		\item 	$m = \max_{x \in \variableset} \multiplicity(x)$, the maximum multiplicity of any variable,
		\item  $t = \max_{T \in \TRSET{C}} \statesize{T}$, the maximum number of states of any transducer, 
		\item $p = \max_{P \in \AUTSET{C}} \statesize{P}$, the maximum number of states (and stack symbols) of any automaton.
		\item $k = \sizeof{\variableset}$, the number of variables. 
	\end{itemize}
	\end{theorem}

Our aim, towards a 2NEXPTIME procedure, is to non-deterministically guess an assignment of size at most double exponential, and check that it satisfies the Conditions 1 and 2 (see Definition~\ref{def:satassign}). However, Condition 2 uses more existentially quantified variables, and it is not evident that verifying Condition 2 can be done within the complexity limits. Towards this, we define an extended assignment which considers the values given to these existentially quantified variables as well, and show that every word used in this extended assignment is of length at most $B$. 
We define the extended assignment and the related notions and notations, and state the small model property for the extended assignment. 

Recall that $\multiplicity(x)$ denotes the number of times a variable occurs on the RHS of the constraints. In each such occurrence, the variable occurs in a pair along with a transducer (of the form $(x, T)$), which belongs to the RHS of a constraint from $\Rel$ of the form $(y,Y)$. Let us fix some enumeration of these occurrences, and define the respective transducer and  constraint of the $i$th occurrence of $x$ by $\Transducer{x}{i}$ and $\constraint{x}{i}$. Now, as per Condition 2, there are words (output words of the respective transducers), that witness the transduction. For every $x \in \variableset$ and $i \in \nset{ \multiplicity(x)}$, let $\outputvar{x}{i}$ be a new variable. This variable is intended to take as value a witness word for the output of the transducer $\Transducer{x}{i}$ on the word provided by $x$, so that the constraint $\constraint{x}{i}$ is satisfied. Let $\extendedvariableset(C)$, or simply $\extendedvariableset$ when $C$ is clear from the context,  contain the output variables in addition to the original variables. That is, $\extendedvariableset = O \cup \variableset$, where $O = \{\outputvar{x}{i} \mid x \in \variableset, i \in \nset{\multiplicity(x)} \}$. 
An extended assignment $\extendedassignment : \extendedvariableset \to \alphabet^\ast$ \textit{satisfies} a string constraint $C$ if
\begin{enumerate}[label={\textbf{E}\arabic*}]
	\item  \label{cond:ext-mem}$\extendedassignment(x) \in \Mem(x)$ for all $x \in\variableset$
	\item  \label{cond:ext-trans}$( \extendedassignment(x), \extendedassignment(\outputvar{x}{i})) \in R(\Transducer{x}{i})$, for all $x\in \variableset$, $i \in \nset{ \multiplicity(x)}$
	\item \label{cond:ext-rel}For every $(y, Y) \in \Rel$,  we have $\extendedassignment(y) \subword  \shuffle{\extendedassignment(Y)}$ where $\extendedassignment(Y)$ is an overloaded  notation  for the multiset 
	\begin{equation}
		\multiset{\extendedassignment(\outputvar{x}{i}) \mid x \in \variableset, i \in \nset{ \multiplicity(x)}, \constraint{x}{i} = (y,Y)}. \label{eqn:sigmahatY}
\end{equation}	
\end{enumerate}

\newcommand{\maxout}{c}
We will actually prove the small model property for the extended assignments. 
\begin{lemma}\label{lemm:extendedsmallmodel}
	Let $C$ be a satisfiable acyclic string constraint. Let $\bound$ be $2^{2^{\mathcal{O}(m\log t+ \log p + \log k)}}$. Then $C$ has a satisfying extended assignment $\extendedassignment$ such that $\len{\extendedassignment(x)} \le \bound $ for all $x \in \variableset(C)$, and $\len{\extendedassignment(y)} \le 2ct\bound$ for all $y \in O$. The parameters $m, t, p, k$ are as defined in Theorem~\ref{thm:smallmodel}, and $\maxout = \max \{ \len{\out(tr)} \mid tr \textrm{ is a transition of } T, \textrm{ and } T \in \TRSET{C} \}$.
	\end{lemma}

 With this, our non-deterministic procedure guesses an extended assignment and verifies that it satisfies the conditions 1, 2 and 3. In fact, checking whether a given word $w$ is a subword of some word in $\shuffle{W}$ where $W$ is a finite mutliset of words is NP-complete \cite{AiswaryaMS22, ChiniKKMS17}.  It remains to prove Lemma~\ref{lemm:extendedsmallmodel}.
 
 \begin{proof}[Proof of Lemma~\ref{lemm:extendedsmallmodel}]

Consider an acyclic string constraint $C$. Recall that we write  $x< y$ if $(x, Y) \in \Rel$ with $(y,T) \in Y$ for some $T \in \TRSET{C}$. Consider a topological sorting of the variables respecting the relation $<$, say $x_1, x_2, \dots , x_k$. Note that $x_1$ does not appear in the RHS of any subword order constraint (in other words, $\multiplicity(x_1) = 0$). If $(x_i, T) $ appears on the RHS of any constraint for some $T$, then the LHS of that constraint is  $x_j$ for some $j < i$. 

Suppose $C$ is satisfiable. Let $\extendedassignment_0$ be a satisfying extended assignment. In order to get the $\extendedassignment$ as per Lemma~\ref{lemm:extendedsmallmodel}, we will construct a sequence of $k$ extended assignments, each progressively modifying the previous one until we reach our goal. 
That is, we will construct the sequence, $ \extendedassignment_0, \extendedassignment_1, \dots ,\extendedassignment_k = \extendedassignment$ such that for each $i \in \nset{k}$
\begin{enumerate}[label=\textbf{I}\arabic*]
	\item\label{cond:inv1}  $\extendedassignment_i$ is satisfiable.  That is, 1) membership constraints are satisfied (Condition~\ref{cond:ext-mem}), 2) transductions are accepted by the transducers (Condition~\ref{cond:ext-trans}),  and 3) the relational constraints are satisfied (Condition~\ref{cond:ext-rel}). 
	\item\label{cond:inv2} for all $j: j \le i $, \, $\len{\extendedassignment_i(x_j)}\le\bound_j$ and for each $\ell \in \multiplicity(x_j)$, $\len{\extendedassignment_i(\outputvar{x_j}{\ell})} 
	\le 2 c t \bound_j$. We define $\bound_n$  as follows. $B_1 = 2^{p^3}$, and  for $n>1$, $B_n = 2m \cdot t^{2m} \cdot p^3 \cdot B_{n-1} \cdot 2^{p^3 t^{2m}}$.
\end{enumerate}

 Note that, $B_n$ increases with $n$ and $\bound_k$ is at most $\bound$.

\noindent \textbf{Base cases }We consider $\extendedassignment_0$ and $\extendedassignment_1$ as base cases. For $\extendedassignment_0$, it is given to be satisfiable, and Condition~\ref{cond:inv2} above holds vacuously. 

Towards constructing $\extendedassignment_1$, consider the variable $x_1$, and a context-free grammar $G_1$ for $\Mem(x_1)$ in Chomsky Normal Form with at most $p^3$ non-terminals \cite{AtigKS13}. Note that $G_1$ can be constructed in polynomial time. Since $\extendedassignment_0$ is satisfying, the word $w_1 = \extendedassignment_0(x_1)$ has a valid parse tree in $G_1$. If a non terminal repeats in any leaf to root path in this tree, say at node $n_1$ and node $n_2$ with $n_2$ an ancestor of $n_1$, then we can shrink the parse tree (pump down) by replacing the subtree  rooted at $n_2$ by the subtree rooted at $n_1$ to get a smaller parse tree of a smaller word in the language. Furthermore, this smaller word will be a subword of $w_1$. Consider a shrinking of the parse tree of $w_1$ which cannot be shrunk any further. This tree has size at most $2^{p^3}$, and hence  its yield $\widehat{w_1}$ satisfies Condition~\ref{cond:inv2}. Setting $x_1$ to $\widehat{w_1}$ will also satisfy Condition~\ref{cond:inv1}. Further, note that there are no output variables corresponding to $x_1$.
Hence we get $\extendedassignment_1$: $\extendedassignment_1(x) = \begin{cases}  \widehat{w_1} &\text{ if }  x = x_1\\
	\extendedassignment_0(x) &\text{ otherwise. }
\end{cases}$

\noindent\textbf{Inductive Step} Now, for the inductive case, assume we have constructed $\extendedassignment_{i-1}$.  We will describe how to obtain $\extendedassignment_{i}$. Let $G_i$ be the context-free grammar in Chomsky Normal Form for $\Mem(x_i)$ with $p^3$ non-terminals. We will basically do a ``conservative'' pumping down of $w_i = \assignment_{i-1}(x_i)$, which ensures that the constraints are still satisfied, which we explain below.

\medskip

\noindent\textbf{Challenges } In order to bound the length of $\extendedassignment_{i}(x_i)$ we may consider subwords $w'_ i \subword w_i$, so that the constraints in which $x_i$ appear on the left continue to be satisfied. In addition, such a subword $w'_i$ must not only satisfy the membership constraint ($w'_i  \in \Mem(x_i)$) but also admit the specified transductions  -- that is, for all $\ell \in \nset{\multiplicity(x_i)}$, we must have $(w'_i, u_\ell') \in R(\transducer_\ell^{x_i})$ for some $u'_\ell$. Furthermore, a mere existence of such a $u'_\ell$ is not sufficient -- consider the constraint $(x_j, Y) = \constraint{x_i}{\ell} $ and let $\extendedassignment_{i-1}(x_j) = w_j$ and $\extendedassignment_{i-1}(Y) = U$ with $\extendedassignment_{i-1}(\outputvar{x_i}{\ell}) = u_\ell$ (recall, $\extendedassignment(Y)$ is defined in Equation~\ref{eqn:sigmahatY})). Since $\extendedassignment_{i-1}$ is satisfying we know that $w_j \subword \shuffle U$. However, it need not be the case that $w_j \subword \shuffle{ U \setminus \{u_\ell\} \cup \{u'_\ell\} }$.  
Hence we need to find a suitable $u'_\ell$ such that $w_j \subword U \setminus \{u_\ell\} \cup \{u'_\ell\}$. One way to ensure this, is by insisting that $u_\ell'$ provides the same ``witnessing subword'' that $u_\ell$ provided. We formalise this notion of ``witnessing subword'' below. 

We give two equivalent definitions for $w \subword \shuffle{U}$.
\begin{claim}\label{claim:equivSAT}
	Let $w$ be a word and $U = \multiset{u_1', u_2' \dots u_n'}$ be a multiset of words. The following statements are equivalent. 
	\begin{enumerate}
		\item There exists $w'$: 1) $w' \in \shuffle{U}$ and 2) $  w\subword w'$. 
		\item There exist $u_1', u_2' \dots u_n'$:  1) $w \in \shuffle{\multiset{u_1', u_2' \dots u_n'}}$  and 2) $u_i' \subword u_i$.
	\end{enumerate}
\end{claim}
\begin{proof}
	Let $w \subword w'$ for some $w' \in  \shuffle{U}$. Then there is an injective map from the positions of $w$ to $w'$. Further since $w' \in  \shuffle{U}$, there is an injective map from $w'$ to positions in the disjoint union of the positions of $u_i$. Composing the two injective map and projecting to the relevant positions in $u_1, u_2, \dots u_n$ will provide us with the required $u_1', u_2' \dots u_n'$such that $w \in \shuffle{\multiset{u_1', u_2' \dots u_n'}}$.
	\medskip
	
	Suppose $u_1', u_2' \dots u_n'$ is such that $w \in \shuffle{\multiset{u_1', u_2' \dots u_n'}}$  and  $u_i' \subword u_i$. We can obtain the required $w'$ by iterating over all $i$ and  inserting into $w$,  the positions in $u_i$ but not in $u'_i$, in order.
\end{proof}

\newcommand{\tikzfigBlockDecom}	{
	\begin{tikzpicture}[node distance=1cm, initial text=,>=latex, thick,]	
		\node[state,inner sep=0pt,minimum size=5pt] (1b0e1) {};
		\node[state,,right of=1b0e1,inner sep=0pt,minimum size=5pt] (1b1r1) {};
		\node[state,,,right of=1b1r1,inner sep=0pt,minimum size=5pt] (1b1r2) {};
		\node[state,,right of=1b1r2,inner sep=0pt,,draw=none,minimum size=5pt] (1b1e1) {};
		
		\node[state,,right of=1b1e1,inner sep=0pt,minimum size=5pt] (1b2r1) {};
		\node[state,,,right of=1b2r1,inner sep=0pt,minimum size=5pt] (1b2r2) {};
		
		\node[state,,right of=1b2r2,inner sep=0pt,,draw=none,minimum size=5pt] (1b3e0) {};
		\node[state,,draw = none,right of=1b3e0,inner sep=0pt,minimum size=15pt] (1d1) {$\cdots$};
		
		\node[state,,right of=1d1,inner sep=0pt,minimum size=5pt] (1b3r1) {};
		\node[state,,,right of=1b3r1,inner sep=0pt,minimum size=5pt] (1b3r2) {};
		
		\node[state,,right of=1b3r2,inner sep=0pt,,minimum size=5pt] (1b4e1) {};
		
		\node[state,below of=1b0e1,node distance=0.65cm,inner sep=0pt,minimum size=5pt] (2b0e1) {};
		\node[state,,right of=2b0e1,inner sep=0pt,minimum size=5pt] (2b1r1) {};
		\node[state,,,right of=2b1r1,inner sep=0pt,minimum size=5pt] (2b1r2) {};
		
		\node[state,,right of=2b1r2,inner sep=0pt,,draw=none,minimum size=5pt] (2b2e1) {};
		\node[state,right of=2b2e1,inner sep=0pt,minimum size=5pt] (2b2r1) {};
		\node[state,,,right of=2b2r1,inner sep=0pt,minimum size=5pt] (2b2r2) {};
		
		\node[state,,right of=2b2r2,inner sep=0pt,,draw=none,minimum size=5pt] (2b3e0) {};
		\node[state,,draw = none,right of=2b3e0,inner sep=0pt,minimum size=15pt] (2d1) {$\cdots$};
		
		\node[state,,right of=2d1,inner sep=0pt,minimum size=5pt] (2b3r1) {};
		\node[state,,,right of=2b3r1,inner sep=0pt,minimum size=5pt] (2b3r2) {};
		
		\node[state,,right of=2b3r2,inner sep=0pt,,,minimum size=5pt] (2b4e1) {};

		\node[state,below of=2b0e1,node distance=1cm,inner sep=0pt,minimum size=5pt] (3b0e1) {};
		\node[state,,right of=3b0e1,inner sep=0pt,minimum size=5pt] (3b1r1) {};
		\node[state,,,right of=3b1r1,inner sep=0pt,minimum size=5pt] (3b1r2) {};
		
		\node[state,,right of=3b1r2,,draw=none,inner sep=0pt,minimum size=5pt] (3b2e1) {};
		\node[state,,right of=3b2e1,inner sep=0pt,minimum size=5pt] (3b2r1) {};
		\node[state,,,right of=3b2r1,inner sep=0pt,minimum size=5pt] (3b2r2) {};
		
		\node[state,,right of=3b2r2,draw=none,inner sep=0pt,minimum size=5pt] (3b3e0) {};
		\node[state,,draw = none,right of=3b3e0,inner sep=0pt,minimum size=15pt] (3d1) {$\cdots$};
		
		\node[state,,right of=3d1,inner sep=0pt,minimum size=5pt] (3b3r1) {};
		\node[state,,,right of=3b3r1,inner sep=0pt,minimum size=5pt] (3b3r2) {};
		
		\node[state,,right of=3b3r2,inner sep=0pt,minimum size=5pt] (3b4e1) {};
		
		\begin{pgfonlayer}{bg}
			\draw[draw=red, fill=red,,fill opacity=0.2, thin, rounded corners ,] ([xshift=-4pt,yshift=6pt]1b0e1.north west) rectangle ([xshift=4pt,yshift=-2pt]1b1r1.south east);		  
			\draw[draw=red, fill=red,fill opacity=0.2, thin,rounded corners  ] ([xshift=-4pt,yshift=6pt]1b1r1.north west) rectangle ([xshift=4pt,yshift=-2pt]1b2r1.south east);		  
			\draw[draw=red, fill=red,,fill opacity=0.2, thin, rounded corners] ([xshift=-4pt,yshift=6pt]1b2r1.north west) rectangle ([xshift=4pt,yshift=-2pt]1b3e0.south east);		
			\draw[draw=red, fill=red,,fill opacity=0.2, thin, rounded corners] ([xshift=-4pt,yshift=6pt]1b3r1.north west) rectangle ([xshift=4pt,yshift=-2pt]1b4e1.south east);
			
			\draw[draw=red, fill=red,,fill opacity=0.2, thin, rounded corners,] ([xshift=-4pt,yshift=6pt]2b0e1.north west) rectangle ([xshift=4pt,yshift=-2pt]2b1r1.south east);		  
			\draw[draw=red,fill=red,,fill opacity=0.2,  thin, rounded corners] ([xshift=-4pt,yshift=6pt]2b1r1.north west) rectangle ([xshift=4pt,yshift=-2pt]2b2r1.south east);		  
			\draw[draw=red, fill=red,,fill opacity=0.2, thin, rounded corners] ([xshift=-4pt,yshift=6pt]2b2r1.north west) rectangle ([xshift=4pt,yshift=-2pt]2b3e0.south east);		
			\draw[draw=red, fill=red,,fill opacity=0.2, thin,rounded corners ] ([xshift=-4pt,yshift=6pt]2b3r1.north west) rectangle ([xshift=4pt,yshift=-2pt]2b4e1.south east);

			\draw[draw=red, fill=red,,fill opacity=0.2, thin,rounded corners ,] ([xshift=-4pt,yshift=6pt]3b0e1.north west) rectangle ([xshift=4pt,yshift=-2pt]3b1r1.south east);		  
			\draw[draw=red, fill=red,,fill opacity=0.2, thin,rounded corners ] ([xshift=-4pt,yshift=6pt]3b1r1.north west) rectangle ([xshift=4pt,yshift=-2pt]3b2r1.south east);		  
			\draw[draw=red, fill=red,,fill opacity=0.2, thin,rounded corners ] ([xshift=-4pt,yshift=6pt]3b2r1.north west) rectangle ([xshift=4pt,yshift=-2pt]3b3e0.south east);		
			\draw[draw=red, fill=red,,fill opacity=0.2, thin,rounded corners ] ([xshift=-4pt,yshift=6pt]3b3r1.north west) rectangle ([xshift=4pt,yshift=-2pt]3b4e1.south east);
		\end{pgfonlayer}
		\draw[->, thin,inner sep=0.5pt, densely dashed] 
		(1b0e1)  -- (1b1r1);
		\draw[->, thin,inner sep=0.5pt, densely dashed] 
		(1b1r2)  -- (1b2r1);
		\draw[->, thin,inner sep=0.5pt, densely dashed] 
		(1b3r2)  -- (1b4e1)
		;
		
		\draw[->, thin,inner sep=0.5pt, densely dashed] 
		(2b0e1)  -- (2b1r1);
		\draw[->, thin,inner sep=0.5pt, densely dashed] 
		(2b1r2)  -- (2b2r1);
		\draw[->, thin,inner sep=0.5pt, densely dashed] 
		(2b3r2)  -- (2b4e1)
		;
		
		\draw[->, thin,inner sep=0.5pt,  densely dashed] 
		(3b0e1)  -- (3b1r1);
		\draw[->, thin,inner sep=0.5pt, densely dashed] 
		(3b1r2)  -- (3b2r1);
		\draw[->, thin,inner sep=0.5pt, densely dashed] 
		(3b3r2)  -- (3b4e1)
		;
		
		\draw[->, thin,inner sep=0.5pt,  densely dashed] 
		(1b2r2)  -- (1b3e0);			
		\draw[->, thin,inner sep=0.5pt,  densely dashed] 
		(2b2r2)  -- (2b3e0);
		\draw[->, thin,inner sep=0.5pt,  densely dashed] 
		(3b2r2)  -- (3b3e0);

		\draw[->, thin,inner sep=0.5pt, ] 
		(1b1r1)  edge node [above] {$a_1$} (1b1r2)
		(1b2r1)  edge node [above] {$a_2$}  (1b2r2)
		(1b3r1)  edge node [above] {$a_n$}  (1b3r2);
		
		\draw[->, thin,inner sep=0.5pt, ] 
		(2b1r1)  edge node [above] {$a_1$} (2b1r2)
		(2b2r1)  edge node [above] {$a_2$}  (2b2r2)
		(2b3r1)  edge node [above] {$a_n$}  (2b3r2);
		
		\draw[->, thin,inner sep=0.5pt, ] 
		(3b1r1)  edge node [above] {$a_1$} (3b1r2)
		(3b2r1)  edge node [above] {$a_2$}  (3b2r2)
		(3b3r1)  edge node [above] {$a_n$}  (3b3r2);
		
		
		\node[state,,,below of=3b0e1,node distance = 0.75cm,  draw=none,inner sep=0pt,minimum size=5pt] (s1) {$\overline{q_0} $};
		\node[state,,,below of=3b1r1,node distance = 0.75cm, draw=none,inner sep=0pt,minimum size=5pt] (s2) {$\overline{q_1} $};
		\node[state,,,below of=3b2r1,node distance = 0.75cm,  draw=none,inner sep=0pt,minimum size=5pt] (s3) {$\overline{q_2} $};
		\node[state,,,below of=3b3r1, node distance = 0.75cm, draw=none,inner sep=0pt,minimum size=5pt] (s4) {$\overline{q_n} $};
		
		\node[state,,xshift=0.5cm,above of=1b1r1,node distance = 0.75cm,  draw=none,inner sep=0pt,minimum size=5pt] (a1) {${a_1} $};
		\node[state,xshift=0.5cm,,above of=1b2r1,node distance = 0.75cm, draw=none,inner sep=0pt,minimum size=5pt] (a2) {${a_2} $};
		\node[state,above of=1d1,node distance = 0.75cm, draw=none,inner sep=0pt,minimum size=5pt] (a25) {$\cdots $};
		
		\node[state,,xshift=0.5cm,above of=1b3r1,node distance = 0.75cm,  draw=none,inner sep=0pt,minimum size=5pt] (a3) {${a_n} $};

		\node[state,,xshift=0cm,left of=1b0e1,node distance = 1cm,  draw=none,inner sep=0pt,minimum size=5pt] (r1) {${\rho_1} $};
		\node[state,xshift=0cm,,left of=2b0e1,node distance = 1cm, draw=none,inner sep=0pt,minimum size=5pt] (r2) {${\rho_2} $};
		\node[state,xshift=0cm,,below of=r2,node distance = 0.4cm, draw=none,inner sep=0pt,minimum size=5pt] (r25) {$\vdots $};
		\node[state,,xshift=0cm,left of=3b0e1,node distance = 1cm,  draw=none,inner sep=0pt,minimum size=5pt] (r3) {${\rho_m} $};
		\node[state,,xshift=0cm,above of=r1,node distance = 0.75cm,  draw=none,inner sep=0pt,minimum size=5pt] (w1) {${w} $};
		
		\draw[draw=blue,  thick,  dashed,rounded corners] ([xshift=-7pt,yshift=8pt]1b0e1.north west) rectangle ([xshift=4pt,yshift=-3pt]s1.south east);		
		\draw[draw=blue,  thick,  dashed,rounded corners,] ([xshift=-7pt,yshift=8pt]1b1r1.north west) rectangle ([xshift=4pt,yshift=-3pt]s2.south east);		  
		\draw[draw=blue,  thick,  dashed,rounded corners,] ([xshift=-7pt,yshift=8pt]1b2r1.north west) rectangle ([xshift=4pt,yshift=-3pt]s3.south east);	
		\draw[draw=blue,  thick,   dashed,rounded corners,] ([xshift=-7pt,yshift=8pt]1b3r1.north west) rectangle ([xshift=3pt,yshift=-3pt]s4.south east);		      	      
		
	\end{tikzpicture}
}
		\newcommand{\overbar}[1]{\mkern 1.5mu\overline{\mkern-1.5mu#1\mkern-1.5mu}\mkern 1.5mu}

\def\Ax{0}
\def\Ay{0}
\def\Bx{-4}
\def\By{-8}
\def\Cx{5}
\def\Cy{\By}
\def\Dx{\Bx}
\def\Dy{\By-2}
\def\Ex{\Cx}
\def\Ey{\Dy}
\def\Fx{\Bx}
\def\Fy{\Cy-1}
\def\Gx{\Cx}
\def\Gy{\Fy}
\def\Hx{\Bx}
\def\Hy{\Fy-4}
\def\Ix{\Cx}
\def\Iy{\Hy}
\def\Mlx{0}
\def\Mly{\By}
\newcommand{\tikzfigParseTree}{
{
	\begin{tikzpicture}
		
		\coordinate (A) at (\Ax,\Ay);
		\node at (A) [anchor = south] {};
		\coordinate (B) at (\Bx,\By);
		\node at (B) [anchor = east] {};
		\coordinate (C) at (\Cx,\Cy);
		\node at (C) [xshift=-25pt,anchor = west] {};
		
		\coordinate (D) at (\Mlx,\By);
		\node at (D) [anchor = south] {};
		\coordinate (E) at (\Mlx+0.7,\By);
		\node at (E) [anchor = south] {};
		\coordinate (F) at ([xshift=12,yshift=-100]A);
		\coordinate (G) at ([xshift=15,yshift=-160]A);
		\coordinate (H) at ([xshift=-50]D);
		\coordinate (I) at ([xshift=-4]D);
		\coordinate (J) at ([xshift=50]D);
		\coordinate (K) at ([xshift=100]D);
		\coordinate (L) at ([xshift=-130]D);
		\coordinate (M) at ([xshift=-112]D);
		\coordinate (N) at ([xshift=142]D);
		
		\draw[color=gray!40,line width = 0.4mm, fill=gray!20] (F) -- (H) -- (I) -- (G) -- (J) -- (K) -- cycle;
		
		
		\fill[black] (F) circle (3pt)node[yshift=-2,left]{\Large$A,\langle \overbar{q_1},\overbar{q_2}\rangle $};
		
		\fill[black] (G) circle (3pt)node[yshift=-2,left]{\Large$A,\langle \overbar{q_1},\overbar{q_2}\rangle $};
		
		
		\fill[black] (M) circle (3pt)node[yshift=-5,below]{\Large$a_{1}$};
		\fill[black] (H) circle (3pt)node[yshift=-5,below]{\Large $a_{i}$};
		\fill[black] (I) circle (3pt)node[yshift=-5,below]{\Large$a_{i'}$};
		\fill[black] (J) circle (3pt)node[yshift=-5,below]{\Large$a_{j'}$};
		\fill[black] (K) circle (3pt)node[yshift=-5,below]{\Large$a_{j}$};
		\fill[black] (N) circle (3pt)node[yshift=-5,below]{\Large$a_{n}$};

		
			\coordinate (D11) at ([xshift=35, yshift=-17]B);
		\node at (D11) [anchor = south] { $\cdots$};

		\coordinate (D11) at ([xshift=85, yshift=-17]B);
		\node at (D11) [anchor = south] { $\cdots$};

		\coordinate (D31) at ([xshift=135, yshift=-17]B);
\node at (D31) [anchor = south] { $\cdots$};

		\coordinate (D41) at ([xshift=185, yshift=-17]B);
		\node at (D41) [anchor = south] { $\cdots$};

		\coordinate (D21) at ([xshift=235, yshift=-17]B);
		\node at (D21) [anchor = south] { $\cdots$};

		\draw[line width = 0.4mm] (A) to[out=-90,in=70] (F) to[out=-100,in=70] (G);
		
		\draw [line width = 0.4mm](A)--(B)--(C)--cycle;
		
	\end{tikzpicture}
}}

We refer to $v_1, \dots v_n$ as the \emph{witnessing subwords} of $u_1, \dots u_n$ for $w \subword \shuffle{U}$.  We can also define the witnessing projection for these witnessing subwords.

\begin{definition}[Witnessing Projection of a subword expression]
	Let $U$ be a multiset of words, and $w$ be a word. Let $u_1, \dots u_n$ be an enumeration of $U$. Let $\pi: \{1,\dots,n\} \to \alphabet^\ast$ be a map. 
	We say that $\pi$ is a \emph{witnessing projection} for the expression  $w \subword \shuffle{U} $  if 
	\begin{enumerate}
		\item ${\pi(i)} \subword u_i$ 
		\item $w \in \shuffle{\{ {\pi(i)} \mid i \in \{1,\dots,n\}\}}$. 
	\end{enumerate}
\end{definition}
\begin{claim}\label{claim:SATiffWitnessProj}
	The expression $w \subword \shuffle{U} $ is satisfiable if and only if it has a witnessing projection. 
\end{claim}
\begin{proof}
	We use the second statement from Claim~\ref{claim:equivSAT}. 
	If the expression is satisfiable, then a witnessing projection $\pi$ will assign $u_i'$ to each $i$. Conversely if there is a witnessing projection, then clearly that gives us the necessary $u_i'$. 
	\end{proof}

Next we extend the notion of witnessing projections to a set of expressions. 
\begin{definition}[Witnessing Projection for a set of expressions]
	Let $R$ be a set of expressions of the form $w \subword \shuffle{U}$. Let $\Pi$ be a map that assigns to each expression $r \equiv w \subword \shuffle{U}$ a map $\pi_r : \{1, 2, \dots \sizeof{U}\} \to \Sigma^\ast$. We say $\Pi$ is a witnessing projection for $R$ if  for each $r \in R$, $\Pi(r)$ is a witnessing projection for $r$.
\end{definition}

\begin{example}
	Consider the expression $baab \subword \shuffle{\multiset{aba, aba}}$. 
	A witnessing projection for this expression assigns $ba$ to the first occurrence and $ab$ to the second occurrence of $aba$ in the RHS.
\end{example}

\noindent
In light of the above definition, we can write Condition~\ref{cond:ext-rel} in the satisfaction of a constraint using extended assignment, equivalently  as
\begin{enumerate}[label = {\textbf{E}\arabic*'}]
	\setcounter{enumi}{2}
	\item\label{cond:ext-rel-dash} There exists a witnessing projection for $\extendedassignment(\Rel)$, where $\extendedassignment(\Rel) = \multiset{ \extendedassignment(r) \mid r\in \Rel}$. Here $\extendedassignment(y, Y) \equiv \extendedassignment(y) \subword \shuffle{\extendedassignment(Y)}$.
\end{enumerate}

Recall that $\extendedassignment(Y)$ is an overloaded  notation  for the multiset $\multiset{\extendedassignment(\outputvar{x}{i}) \mid x \in \variableset, i \in \{1, 2, \dots, \multiplicity(x)\}, \constraint{x}{i} = (y,Y)}.$ Suppose $\Pi$ is a witnessing projection for $\extendedassignment(\Rel)$. Let $r \equiv (x, Y) \in \Rel$ and  let $\Pi(r) = \pi_r$. Consider the enumeration of $\extendedassignment(Y) = \multiset{\extendedassignment(o_1), \extendedassignment(o_2), \dots} $ where each $o_j$ is an output variable. Instead of $\pi_r(i)$ we may also write $\pi_r(o_i)$.

Thus, for the inductive case, we need to find good $w'_i$, $u_1', \dots u'_{\multiplicity(x_i)}$ such that
\begin{enumerate}
	\item $\sizeof{w'_i} \le B_i$
	\item  for each $\ell$, $\sizeof{u'_\ell} \le 2 \cdot t \cdot B_i$,
	\item $w'_i  \in \Mem(x_i)$,
	\item $(w'_i, u_\ell') \in R(\transducer_\ell^{x_i})$ for each $\ell$,  and 
	\item $v_\ell \subword u'_\ell$ where $v_\ell$ is a witnessing subword of $u_\ell = \extendedassignment_{i-1}(\outputvar{x_i}{\ell})$ for the relation  $\extendedassignment_{i-1}(y) \subword  \shuffle{\extendedassignment_{i-1}(Y)}$ letting $(y, Y) = \constraint{x_i}{\ell}$. (Note that $v_\ell$ exists because $\extendedassignment_{i-1}$ is satisfiable by induction hypothesis.)
\end{enumerate}

\noindent\textbf{Block decomposition } Towards the above goal, let us consider $w_i$ and $\rho_1 \dots \rho_m$ where $m = \multiplicity(x_i)$ and for $\ell: 1 \le \ell \le m$,  $\rho_\ell$ is an accepting run of the transducer $\Transducer{x_i}{\ell}$ on the input $w_i$ producing $u_\ell = \extendedassignment_{i-1}(\outputvar{x_i}{\ell})$. This is depicted in Figure~\ref{fig:blockdecompose}. We factorize  each $\rho_\ell$ into blocks. The number of blocks is exactly $n+1$, where $ n = \sizeof{w_i}$. For $j>0$, the $j$th block contains the transition on the $j$th letter of $w_i$, followed by all the trailing $\epsilon$-input transitions. The very first block (block $0$),  contains the leading $\epsilon$-input transitions if present. Now, we decompose the output of $\rho_\ell$ according to the blocks. That is, $u_\ell = u^0_\ell u^1_\ell \cdot  u^2_\ell \dots  u^n_\ell  $. Next we want to identify the subword of $w_i$ (and subruns of $\rho_\ell$) that needs to be preserved. 

	\begin{figure}
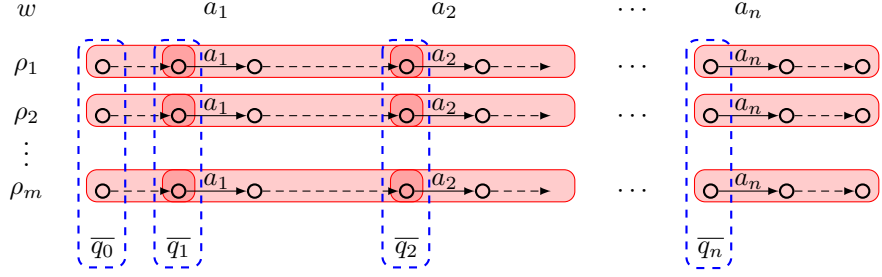

		\centering
	\scalebox{1}{\tikzfigBlockDecom}

	\caption{The figure describes the block decomposition of the runs of transducers. The very first block includes the sequence of $\epsilon$-transitions (if any) before the first letter of the input is read. Every other block includes a transition on input letter followed by a sequence of $\epsilon$ transitions.  Blocks here are represented as overlapping rectangles coloured in red.  The transition on input is represented as a solid arrow and a sequence of $\epsilon$ transitions is represented as a dashed arrows.} \label{fig:blockdecompose}
	\end{figure}

\medskip

\noindent\textbf{Identifying crucial blocks and positions.} Consider $v_\ell$, the witnessing subword of $u_\ell$ for $\extendedassignment_{i-1}(y) \subword  \shuffle{\extendedassignment_{i-1}(Y)}$  where $(y, Y) = \constraint{x_i}{\ell}$. Fix an embedding of $v_\ell$ in $u_\ell$. If this embedding is incident\footnote{This means that the image of the witnessing projection intersects the positions of the factor $u^j_\ell$.} on the factor $u^j_\ell$ for $j>0$, we will mark the the $j$th block as well as the $j$th letter of $w_i$ as \emph{crucial}. 
Since $\sizeof{\extendedassignment_{i-1}(y) } \le B_{i-1}$ (by induction hypothesis), the number of crucial blocks in $\rho_\ell$ is at most $B_{i-1}$. Hence the number of crucial positions in $w_i$ is at most $m \times B_{i-1} $ where $m = \multiplicity(x_i)$. Notice that if we  shrink $w_i$ to a subword that 1) preserves the crucial positions, 2) preserves membership in $\Mem(x_i)$ and 3) yields subruns of $\rho_\ell$ that preserves the crucial blocks, then the satisfiability would be preserved. Our next aim is to obtain such a shrinking, which is sufficiently small to also satisfy the length requirements.

\noindent\textbf{Annotated parse trees}  Consider a grammar $G_i$ for $\Mem(x_i)$ in Chomsky Normal Form and a parse tree of $w_i$ in $G_i$. Annotate the nodes of this parse-tree by pairs of $m$-tuple of states. The $m$-tuple of states $\overline{q_{j}}$ correspond to the states of the transducers at the boundary between $(j-1)$th block  and $j$th block. A node is annotated with  $\langle\overline{q_{j-1}}, \overline{q_{j'}}\rangle$ if the yield of the subtree rooted at that node generates the factor of $w_i$ from $j$th letter to $j'$th letter (for some $j' \ge j$). Notice that some of the leaves are marked as crucial. We will mark an internal node as crucial if it is the least common ancestor of two crucial nodes. 

\begin{figure}
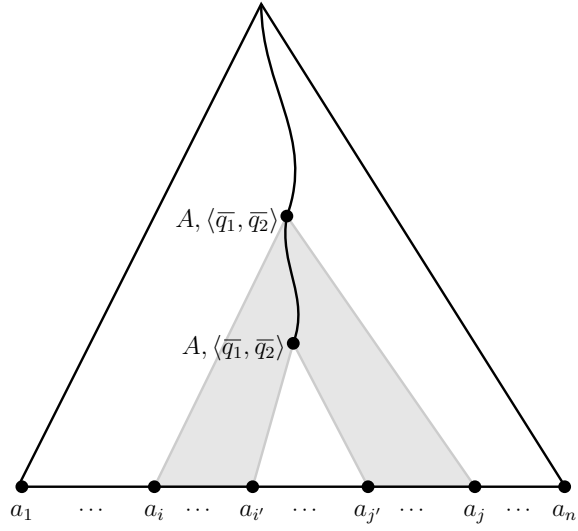

	\centering
	\scalebox{0.8}{\tikzfigParseTree}
	\caption{The figure illustrates the annotations of a nodes and pumping down in a parse tree}\label{fig:parsetree}
\end{figure}

\noindent\textbf{Shrinking the parse tree}
Now, if there are two nodes $n_1$ and $n_2$ in this tree such that 1) both  have the same annotated non-terminal, 2) $n_1$ is an ancestor of $n_2$, 3) there are no crucial nodes in the path from $n_1$ to $n_2$, then we replace the subtree rooted at $n_1$ with the subtree rooted at $n_2$ (pumping down). This is illustrated in Figure~\ref{fig:parsetree}. We repeat this until no more pumping down is possible. The yield of this shrunk parse tree is the required word $  {w'_i}$. Let us analyse the size of ${w'_i}$. 
 Any path without a crucial node is of length at most  $p^3 t^{2m} $. Hence the \emph{skeleton} of the parse tree that contains all the crucial nodes and the paths from them to the root will be of size at most $2  n_C  \times p^3 t^{2m}$, where $n_C$ is the number of crucial positions of $w_i$. 
 
  Any sub tree rooted at any of the nodes of the skeleton is of size at most $ 2^{p^3 t^{2m}} $. Hence the total size of the tree is at most $2  n_C  \times p^3 t^{2m} \times 2^{p^3  t^{2m}} $. Since $n_C \leq m  B_{i-1}  $, we have $\len{w'_i} \leq B_i $.

\noindent\textbf{Shrinking the transducer runs } Note that, since the shrinking preserves the annotations,  shrinking the $\rho_\ell$s appropriately   gives us an accepting subrun $\rho'_\ell$ that preserves the crucial blocks. The number of blocks in $\rho'_\ell$ is at most $B_i$. Now, we need to shrink the size of each block as well, in order to satisfy $\len{u'_\ell} \le 2ctB_i$. For this, consider the witnessing subword of $u_\ell$. Note that it is still embedded in $\out(\rho'_\ell)$. If this embedding is incident on the output of a transition we will mark this transition as crucial. Further all the transitions that read a letter from $w'_i$ are also crucial. Note that the number of crucial transitions is at most $B_i + B_{i-1}$. Now, let us shrink the run $\rho'_\ell$ without losing crucial transitions to get $\rho''_\ell$. The number of transitions in $\rho''_\ell$ is at most $t \times (B_i + B_{i-1} +1)$ where $t$ is the number of states. Let $u_\ell ' = \out(\rho''_\ell)$. Then it is easy to see that   $\len{u'_\ell} \le 2ctB_i$.

Finally, we can give the required $\extendedassignment_i$. Below, $x_j$ comes from $\variableset$ and $\outputvar{x}{\ell}$ comes from $O$. 

$\extendedassignment_i(x_j) = \begin{cases}  {w'_i} &\text{ if }  j = i\\
	\extendedassignment_{i-1}(x_j) &\text{ if } j \neq i
\end{cases}$
\hfil
$\extendedassignment_i(\outputvar{x}{\ell}) = \begin{cases}  {u'_\ell} &\text{ if }  x = x_i \\
	\extendedassignment_{i-1}(\outputvar{x}{\ell}) &\text{ otherwise}
\end{cases}$

This establishes the proof of Lemma~\ref{lemm:extendedsmallmodel}.
 \end{proof}

\section{ Satisfiability in NEXPTIME for regular constraints} \label{sec:reg}
Our approach towards an \nexp\ procedure is very similar to that of the previous section.  Towards this, we prove the following lemma that shows that if there is a satisfying assignment, then there is a satisfying extended assignment of at most exponential size. Further more,  the proof of the lemma is very similar to the proof of Lemma~\ref{lemm:extendedsmallmodel}.  We only highlight the main differences with it here.
\newcommand{\regbound}{D}
\begin{lemma}\label{lemm:regsmallmodel}
	Let $C$ be a satisfiable acyclic regular string constraint. Then $C$ has a satisfying extended assignment $\extendedassignment$ such that $\len{\extendedassignment(x)} \le \regbound $ for all $x \in \extendedvariableset(C)$, where  $\regbound$ is ${2^{\mathcal{O}(m\cdot k\cdot\log t) }}$ where 
		  $t = \max_{T \in \TRSET{C} \cup  \AUTSET{C}} \statesize{T}$ is the maximum number of states of any transducer or the NFA occurring in $C$, and
		  $m, k$ are as in Theorem \ref{thm:smallmodel}.
\end{lemma}

 As in the previous section, to obtain a satisfying extended assignment of exponential size,  we will construct a sequence of $k$ extended assignments  $ \extendedassignment_0, \extendedassignment_1, \dots ,\extendedassignment_k$ such that for each $i \in \{1,\dots,k\}$, $\extendedassignment_i$ is a satisfying extended assignment  and further for each $j \leq i$, $\extendedassignment_{i}(j) \leq D_i$, where $D_i = {2^{((m+1)\cdot i\cdot\log t + i \cdot \log m)}}$. The base case is immediate,   the smallest assignment for the variable $x_1$ is of size at most $t$. 
For the inductive case, we assume that we have already constructed the extended assignment $\extendedassignment_{i-1} $ of appropriate size. 
Consider  $w_i = a_1 \dots a_n =  \extendedassignment_{i-1}(x_i)$ and $\rho_1 \dots \rho_m$ where $m = \multiplicity(x_i)$ and for $\ell: 1 \le \ell \le m$,  $\rho_\ell$ is an accepting run of the transducer $\Transducer{x_i}{\ell}$ on the input $w_i$ producing $u_\ell = \extendedassignment_{i-1}(\outputvar{x_i}{\ell})$.  Further let $\rho$ be a run in $\Mem(x_i)$ on $w_i$. 
Now consider the block decomposition of each $\rho_i$, $i \in \{1,\dots,m\}$ as in the previous section and mark  the crucial blocks based on the witnessing subwords.  
Notice that there are at most $m \times D_{i-1} $ crucial blocks. We annotate each input letter with the $m$-tuple of states that appears in the corresponding block boundary, 
as in the previous section.
Further, we mark the input letter as crucial if the corresponding block is crucial. 
Now for any two input letters $a_i,a_j$, if the annotations ($m$ tuple of states) for it are the same, there are no crucial letters between them and the states reached in $\rho$ after reading it is the same, then we delete all the transitions between them. Further, we also delete the corresponding blocks. Let $\rho'$ be a sub-run of $\rho$ and $(\rho'_j)_{j\in \{ 1, \dots m\}}$ be  sub-runs of $(\rho_j)_{j\in \{ 1, \dots m\}}$ such that no more deletions are possible. Firstly notice that each of these are a valid runs in the respective automata.  We claim that $\len{\out(\rho'_0)} \leq D_i$, for this we note that there can be at most $m \times D_{i-1}  $ many crucial letters and between any two of them there can be at most $t^{m+1}$ many letters. With this we obtain that $\len{\out(\rho')} \leq (m \times D_{i-1} \times  t^{m+1}) \leq D_i$.
 The length of each $\rho'_j$ for $j \in \{ 1, \dots,m\}$ can still be very large. An analysis similar to the one done in the previous section that shrinks any long sequence of $\epsilon$ transitions within each block,  will also provide us with the required bounds for the variables from $O$.

\section{Discussions} \label{sec:discussions}

\subsection{Application: Regular abstractions and DFA sizes }

Let $L$ be any language.  We define the Parikh image closure ($\Pi(L)$), downward closure ($L\downarrow$) and upward closure ($L\uparrow$) of it as follows. Let $\Sigma = \{a_1, \cdots a_n \}$.  For any word $w \in \Sigma^*$, we let $p(w) = \langle \len{w}_{a_1}, \cdots, \len{w}_{a_n}  \rangle$ denote the Parikh image of $w$, that is, it counts the occurrences of each letter from $\Sigma$. Here, by $\len{w}_{a}$, we mean the number of times $a$ occurs in $w$. 
\[
\begin{array}{lr}
  \Pi(L) = \{v \in \Sigma^* \mid \exists w \in L, p(v) = p(w) \} \qquad&\qquad L\downarrow = \{v \in \Sigma^* \mid \exists w \in L, v \subword w\} \\
  L\uparrow = \{v \in \Sigma^* \mid \exists w \in L, w \subword v\} 
\end{array}
\]

Efficient computability of these regular abstractions of languages of infinite state systems is a relevant question for verification and automata theory \cite{AtigCHKSZ16, AtigBKS14, Zetzsche16,AtigBT08}. It is interesting to see if small automata representing these abstractions can be computed for succinctly given infinite state systems.

{	We address here the case where a large pushdown system is presented as a small pushdown system and a certain number of finite state automata (referred to as the smaller components). Here the language of the large pushdown system is same as the intersection of the languages of the smaller component.  We argue that the lower bound on the size of the regular abstraction  holds even when the language of a pushdown system is presented succinctly as an intersection of smaller components. 

	\medskip

For any $n \in \Nat$, let $L(n)$ be the language over $\Sigma = \{0,1,a\}$ that accepts the word $w = n_0o_1n_1o_2n_2 \cdots n_k$, where $n_0 = 0^n$, $n_k = 1^n$, $o_1,o_2, \cdots o_k = \inc$, then $\len{w}_{\inc} = 2^n$. From Section~\ref{sec:2nexp-lower-bound} we know that we can construct $n$ DFAs $B_1, B_2, \dots B_n$ such that $\bigcap_i L(B_i) = \{w\}$. Since any DFA recognising the closure of this language requires at least $2^n$ states, we have the following claim.

\begin{claim}
Given $n$ regular languages as $n$ finite state automata, let $L$ be the language obtained by intersecting the languages of these automata. Then, the regular representations for Parikh image closure, downward closure and upward closure of $L$ can be of exponential size.
\end{claim}

Consider the language given in Claim~\ref{claim:countingPDA} i.e. $L(A) \cap L'_\ell$, since it can recognize words with exactly $2^{2^n}$ many symbols from $\Gamma_2$, we have the following claim.

\begin{claim}
	Given $n$ regular languages as $n$ finite state automata and a pushdown system, let $L$ be the language got by intersecting the languages of these automata. Then,  the regular representations for Parikh image closure, downward closure and upward closure of $L$ can be of double exponential size.
\end{claim}

\subsection{Concatenation instead of Shuffle}
In \cite{AiswaryaMS22} it is shown that shuffle can express concatenation with a polynomial blow-up, but preserving acyclicity. It is interesting to see if the hardness holds in the presence of concatenation alone. 
Already, in the setting of \cite{AiswaryaMS22} (no transductions), if acyclicity is not imposed, it is not known whether satisfiability of the regular string constraints is decidable if only concatenation is allowed instead of shuffle. In our setting (in the presence of tranducers), it turns out that satisfiability  is undecidable. We show this by modifying the reduction in \cite{AiswaryaMS22}.

\newcommand{\Lang}{\mathcal{L}}
\newcommand{\Aut}{\mathcal{A}}
\newcommand{\Vars}{\mathcal{V}}
Let $\pcp=(\Sigma_1, \Sigma_2, f, g)$ be a given PCP instance,  let $\Lang_{u} = (\{i \cdot f(i) \mid i \in \Sigma_1 \})^*$, $\Lang_{v} = (\{i \cdot g(i) \mid i \in \Sigma_1 \})^*$, $\Lang_i = \Sigma_1^+$ and $\Lang_s =  \Sigma_2^*$.
Notice that all these four languages are regular,  let $\Aut_{u}, \Aut_{v}, \Aut_{i}$ and $\Aut_{s}$ be their corresponding NFA.  Then the required string constraint is $ \tuple{\alphabet, \variableset,  \Mem, \Rel} $, where  $\Vars = \{ u,v,i,s \}$, for any $x \in \Vars$, $\Mem(x) = \Aut_x$.
Let $\transducer_\alphabet = \{(a_1 \dots a_n, \alphabet^*a_1\alphabet^*a_1 \dots \alphabet^* a_n \alphabet^*) \}$ be the transduction that arbitrarily inserts words from $\alphabet$. Then the set $\Rel$ is given by
\begin{align*}
	i \subword u    &&  s \subword u             &	 &               u \subword \transducer_{\Sigma_1}(s)  &&  u \subword \transducer_{\Sigma_2}(t) \\
	i \subword v           &&	  s \subword v       &    &       v \subword \transducer_{\Sigma_1}(s)  &&  v \subword \transducer_{\Sigma_2}(t) 
\end{align*}

In the case of acyclic  string constraints, one may wonder if the lower bounds hold in a setting where only concatenation is allowed instead of shuffle. This was not discussed in \cite{AiswaryaMS22}. Infact, in our case,  the lower bound holds even when only concatenation is allowed. Notice that, in our \twonexph\ reduction, all relational constraints have only one variable in the RHS. Hence, the \twonexph\ holds for acyclic pushdown string constraints with transducers, even when shuffle (or even concatenation) is disallowed. 

In fact, it is not known whether the lower bounds in \cite{AiswaryaMS22} hold for the variant with only concatenation instead of shuffle.  It is also open whether satisfiability is decidable for the unrestricted regular string constraints   (without acyclicity restriction) when only concatenation is allowed.

\section{Conclusions}

In this paper, we considered string constraints in the presence of sub-word relation, shuffle operator (which subsumes concatenation \cite{AiswaryaMS22}) and transducers. We studied this problem for two different kinds of membership constraints, namely regular and context free. We showed that in the case when only regular membership constraints are involved, the problem is \nexpc. Whereas, when context-free membership constraints are involved, the problem is \twonexpc. Towards the hardness proof, we showed how to count exactly $2^n$ using $n$ finite state automata each of size  $\mathcal{O}(n)$. As a consequence of this result, we also obtained a  lower bound for any regular representation of the upward closure, downward closure and Parikh image closure of the intersection of the language of $n$ finite state automata.
Similarly, we showed that we can count exactly $2^{2^n}$ using a pushdown automaton and $n$ finite state automata, each of size $\mathcal{O}(n)$.  
 With this, we also obtained a  lower bound for any regular representation of the upward closure, downward closure and Parikh image closure of the intersection language of  a pushdown  and $n$ finite state automata.
 
 \bibliography{subwordtransducer}

\end{document}